\documentclass[a4paper]{article}
\usepackage{geometry}
\usepackage{amsmath}
\usepackage{amssymb}
\usepackage{indentfirst}
\usepackage{mathdots}
\usepackage{cite}
\usepackage[colorlinks, linkcolor=blue, citecolor=blue]{hyperref}
\usepackage{xcolor}
\usepackage{mathrsfs}
\usepackage{booktabs}
\usepackage{enumitem}
\usepackage{authblk}
\usepackage{afterpage}
\usepackage{amsfonts,ntheorem}
\usepackage{threeparttable}
\usepackage{bm}
\usepackage{subfigure, epsfig}
\usepackage{setspace}
\usepackage{makecell}
\numberwithin{equation}{section}

\usepackage{algorithm}
\usepackage{algpseudocode}
\usepackage{comment}
\usepackage{cases}

\makeatletter

\newcommand{\Rmnum}[1]{\expandafter\@slowromancap\romannumeral #1@}
\makeatother

\newcommand\rank{\operatorname{rank}}
\newtheorem{theorem}{Theorem}[section]
\newtheorem{proposition}[theorem]{Proposition}
\newtheorem{lemma}[theorem]{Lemma}
\newtheorem{corollary}[theorem]{Corollary}

{ 
\theoremstyle{plain}
 
\theorembodyfont{\normalfont} 

\newtheorem{definition}[theorem]{Definition}
\newtheorem{example}[theorem]{Example}
\newtheorem{remark}[theorem]{Remark}

}

\newenvironment{proof}{\noindent{\textbf{\emph{Proof.}}}}

\allowdisplaybreaks[4]

\definecolor{mydarkgreen}{RGB}{0,100,0}

\begin{document}
\title{\Large Optimal Quantum $(r,\delta)$-Locally Repairable Codes via Classical Ones}
\author[1]{{\small{Kun Zhou}} \thanks{E-mail address: \texttt{kzhou@bimsa.cn}}}
\author[1]{{\small{Meng Cao}} \thanks{E-mail address: \texttt{mengcaomath@126.com} (corresponding author)}}
\affil[1]{\footnotesize{Beijing Institute of Mathematical Sciences and Applications, Beijing, 101408, China} }
\renewcommand*{\Affilfont}{\small\it}

\date{}
\maketitle

{\linespread{1.4}{

\vskip -7mm

\noindent {\small{{\bfseries{Abstract:}}
Locally repairable codes (LRCs) play a crucial role in mitigating data loss in large-scale distributed and cloud storage systems. This paper establishes a unified decomposition theorem for general optimal $(r,\delta)$-LRCs.
Based on this, we obtain that the local protection codes of general optimal $(r,\delta)$-LRCs are MDS codes with the same minimum Hamming distance $\delta$.
We prove that for general optimal $(r,\delta)$-LRCs, their minimum Hamming distance $d$ always satisfies $d\geq \delta$. 
We fully characterize the optimal quantum $(r,\delta)$-LRCs induced by classical optimal $(r,\delta)$-LRCs that admit a minimal decomposition. We construct three infinite families of optimal quantum $(r,\delta)$-LRCs with flexible parameters.
}}

\vspace{6pt}


\vspace{6pt}
}

\section{Introduction}\label{section1}
The study of locally repairable codes (LRCs) is of great importance due to their critical role in ensuring efficient data recovery for distributed storage systems. A code is classified as an optimal LRC if its parameters satisfy the Singleton-like bound, which establishes the theoretical limit for balancing redundancy and fault tolerance. Optimal LRCs are particularly valuable as they achieve the best possible tradeoff between storage overhead and repair efficiency, making them highly desirable for deployment in large-scale distributed storage applications.

The theoretical framework of classical LRCs originated in \cite{Gopalan2012} and was subsequently generalized in \cite{Prakash2012}, where the corresponding Singleton-like bound for these codes was formally established. Codes attaining this bound on minimum distance are designated as optimal $(r,\delta)$-LRCs, representing the fundamental class of codes that optimally balance redundancy and repair efficiency. For example, pyramid codes constitute a class of optimal LRCs with flexible parameters \cite{Huang2013, Prakash2012}. In recent years, optimal $(r,\delta)$-LRCs have attracted extensive research interest (see  \cite{Song2014, Kong2021, Luo2022, Hao2022, Prakash2012, Chen2023, Zhu2025, Cai2021} for example). Notably, Song et al. \cite[Theorem 9]{Song2014} established a landmark structural decomposition theorem, proving that under specific conditions, optimal LRCs can be expressed as disjoint unions of maximum distance separable (MDS) codes. This work motivates us to establish a unified decomposition theorem for general optimal $(r,\delta)$-LRCs, where the key challenge lies in developing a unified approach that addresses new cases (e.g., when $r|(k-1)$).

The quantum counterpart of LRCs emerged in 2023 through the seminal work of Golowich and Guruswami \cite{Golowich2023}, who pioneered the first quantum adaptation of LRCs, thereby opening a new research direction in quantum error correction. Their work formulated quantum LRCs and investigated potential applications in quantum data storage architectures. Building upon this foundation, Luo et al. \cite{Luo2025} conducted fundamental studies on quantum LRCs, deriving theoretical bounds and developing construction methodologies for optimal quantum LRCs using CSS frameworks and classical coding-theoretic constraints. In their paper \cite{Sharma2025}, Sharma et al. develop a quantum LRCs construction method capable of employing arbitrary good polynomials, and propose a novel approach for designing such polynomials using subgroups of affine general linear groups.

In 2024, Galindo et al. \cite{Galindo2024} extended quantum LRCs to the more general framework of quantum $(r,\delta)$-LRCs. Their seminal work investigated both fundamental properties and concrete constructions of quantum $(r,\delta)$-LRCs.
They established a remarkable equivalence principle: when a quantum stabilizer code originates from Hermitian dual-containing or Euclidean dual-containing codes under specific additional conditions, there exists an equivalence between classical and quantum notions of $(r,\delta)$-local recoverability.
Building upon this important foundation, we will demonstrate that these specific additional conditions are automatically satisfied in the optimal case, which involves the use of our unified decomposition theorem. Consequently, the CSS and Hermitian constructions remain direct bridges for converting classical optimal $(r,\delta)$-LRCs into their quantum optimal counterparts.

By leveraging our unified decomposition theorem as a new tool, we reveal that any local protection code (see Definition \ref{prote}) of an optimal
$(r,\delta)$-LRC must be an MDS code with minimum Hamming distance $\delta$. This property imposes strong constraints on the selection of optimal $(r,\delta)$-LRCs, leading to intrinsic structural phenomena that eliminate the special conditions previously required in \cite[Theorem 30]{Galindo2024}.

Beyond the seminal work of Galindo et al. \cite{Galindo2024}, systematic study of quantum $(r,\delta)$-LRCs via classical $(r,\delta)$-LRCs, particularly optimal ones with $\delta \geq 3$, remain scarce. A key reason could stem from the CSS and Hermitian constructions imposing overly restrictive conditions on classical codes. However, our new tool, i.e., the unified decomposition theorem, not only enables sufficient construction of optimal quantum $(r,\delta)$-LRCs, but also allows us to fully characterize those satisfying specific conditions. To the best of our knowledge, a complete characterization of certain optimal quantum $(r,\delta)$-LRCs remains unexplored in prior literature. Finally, to demonstrate the diversity of our fully characterized optimal quantum $(r,\delta)$-LRCs, we explicitly construct three infinite families of these codes with distinct and flexible parameters.

The main results of this paper include the decomposition theorem for optimal $(r,\delta)$-LRCs (see Theorem \ref{Stru}), and the following two key theorems:

\begin{theorem}\emph{(see also Theorem \ref{stru-coro1})}
Every local protection code of an optimal $(r,\delta)$-LRC is an MDS code with parameters of the form $[n\leq (r+\delta-1),n+1-\delta,\delta]_{q}$.
\end{theorem}

\begin{theorem}\emph{(Optimal Quantum $(r,\delta)$-LRCs Construction, see also Theorem \ref{stru-corobound})}
For an optimal $(r,\delta)$-LRC with parameters $[n,k,d]_{q^{2}}$ (resp. $[n,k,d]_{q}$), the following statements hold:
\begin{itemize}
\item [(1)] $n-k\geq \lceil \frac{k}{r}\rceil(\delta-1)$, i.e., $d\geq \delta$;

\item [(2)] If it is Hermitian dual-containing (resp. Euclidean dual-containing), then the induced quantum code is an optimal quantum $(r,\delta)$-LRC.
\end{itemize}
\end{theorem}
Here, ``induced quantum codes'' refer to quantum codes obtained via Hermitian construction or CSS construction. Furthermore, we establish the following theorem.

\begin{theorem}\emph{(see also Theorem \ref{Stru-quan})}
If $\mathcal{C}$ is an optimal $(r,\delta)$-LRC that admits a minimal decomposition, and it induces an optimal quantum $(r,\delta)$-LRC, then $\mathcal{C}$ possesses a parity-check matrix of the form in Eq. \eqref{pari} with the property that all codes generated by the matrices in Eq. \eqref{hermi} are either all Hermitian self-orthogonal or all Euclidean self-orthogonal.

Conversely, if $\mathcal{C}$ has such a parity-check matrix of the form in Eq. \eqref{pari} satisfying that all codes generated by the matrices in Eq. \eqref{hermi} are either all Hermitian self-orthogonal or all Euclidean self-orthogonal, then $\mathcal{C}$ admits a minimal decomposition and induces an optimal quantum $(r,\delta)$-LRC.
\end{theorem}
\vspace{-4pt}
Here, ``minimal decomposition'' can be seen in Definition \ref{mini-dec}, which is guided by the unified decomposition theorem. The parity-check matrix in this theorem are fully characterized and computable, thereby completely determining optimal quantum $(r,\delta)$-LRCs induced by these optimal $(r,\delta)$-LRCs.

\vspace{10pt}

Additionally, we present three infinite families of optimal quantum $(r,\delta)$-LRCs in Theorems \ref{Opt-1}, \ref{Opt-2}, and \ref{Opt-3}, respectively.

\vspace{4pt}

This paper is organized as follows:
\vspace{-6pt}
\begin{itemize}
\item  Section \ref{sec:pre} covers essential background on classical codes and quantum codes.

\vspace{-4pt}

\item  Section \ref{OD} has two parts:
\vspace{-4pt}
    \begin{itemize}
\item  Subsection \ref{sub3.1} proves our key decomposition theorem for optimal $(r,\delta)$-LRCs (see Theorem \ref{Stru});
\vspace{-4pt}

\item  Subsection \ref{sub3.2} applies this theorem to analyze local protection codes.
    \end{itemize}

\vspace{-4pt}

\item  Section \ref{quantum} removes the technical requirement in Lemma \ref{quanum-cl}, enabling our construction of optimal quantum $(r,\delta)$-LRCs. We then establish a systematic framework for building these codes.

\vspace{-4pt}

\item  Section \ref{construction} presents three infinite families of optimal quantum $(r,\delta)$-LRCs with explicit parameters.
\end{itemize}

\section{Preliminaries}\label{sec:pre}
Let $\mathbb{F}_{q}$ denote the finite field of order $q$, where $q$ is a prime power. For any matrix $A=(a_{i,j})$ over $\mathbb{F}_{q}$, let $A^{\top}=(a_{j,i})$ denote its {\it transpose}. We employ the following conventions:

\begin{itemize}
\item For any set $S$, let $|S|$ denote its {\it cardinality}.

\item For any positive integer $n$, define $[n]:=\{i\in \mathbb{N}^+|\;1\leq i\leq n\}$.

\item All vectors in $\mathbb{F}_q^k$ are assumed to be column vectors by default and are denoted using boldface notation (e.g., $\mathbf{c}$), unless explicitly stated otherwise.

\item For any subset $T=\{\mathbf{c}_1,\ldots,\mathbf{c}_i\}\subseteq \mathbb{F}_q^k$,
    \begin{itemize}
\item  $\text{span}(T):=\langle \mathbf{c}_1,\ldots,\mathbf{c}_i\rangle $ denotes the linear span of $T$;
\item  $\rank(T):=\dim(\text{span}(T))$.
    \end{itemize}

\item $\emptyset$ represents the empty set.

\item For an $[n,k,d]_q$ linear code $\mathcal{C}$ with a generator matrix $G=(\mathbf{c}_1,\ldots, \mathbf{c}_n)$, we canonically associate the code with an indexed set, denoted by the capital letter $C$, defined as
    \begin{align}
    C:=\{(1,\mathbf{c}_1)\ldots, (n,\mathbf{c}_n)\}\subseteq [n]\times \mathbb{F}_{q}^k. \label{cano}
    \end{align}
\end{itemize}

\subsection{Linear codes}
A {\it linear code} over $\mathbb{F}_{q}$ of length $n$, dimension $k$ and minimum (Hamming) distance $d$ is denoted as $[n,k,d]_{q}$. Such a code forms a $k$-dimensional subspace of the vector space $\mathbb{F}_{q}^{n}$. By the {\it Singleton bound}, the minimum distance satisfies $d\leq n-k+1$. When the equality $d=n-k+1$ holds, the code is called a {\it maximum distance separable (MDS) code}.

For a linear code $\mathcal{C}$ of length $n$ over $\mathbb{F}_{q}$ (resp. $\mathbb{F}_{q^{2}}$), its {\it Euclidean dual} $\mathcal{C}^{\perp_{\mathrm{E}}}$ (resp. {\it Hermitian dual} $\mathcal{C}^{\perp_{\mathrm{H}}}$) are respectively defined by
\begin{align*}
\mathcal{C}^{\perp_{\mathrm{E}}}=\Big\{\mathbf{x}=(x_{1},\ldots,x_{n})\in\mathbb{F}_{q}^{n}:\
\langle\mathbf{x},\mathbf{y}\rangle=\sum_{i=1}^n x_i y_i=0 \ \mathrm{for} \ \mathrm{all} \  \mathbf{y}=
(y_{1},\ldots,y_{n})\in \mathcal{C}\Big\}
\end{align*}
and
\begin{align*}
\mathcal{C}^{\perp_{\mathrm{H}}}=\Big\{\mathbf{x}=(x_{1},\ldots,x_{n})\in\mathbb{F}_{q^{2}}^{n}:\
\langle\mathbf{x},\mathbf{y}\rangle_{\mathrm{H}}=\sum_{i=1}^n x_i y_i^{q}=0 \ \mathrm{for} \ \mathrm{all} \  \mathbf{y}=
(y_{1},\ldots,y_{n})\in \mathcal{C}\Big\}.
\end{align*}
If $\mathcal{C}\subseteq\mathcal{C}^{\perp_{\mathrm{E}}}$ (resp. $\mathcal{C}\subseteq\mathcal{C}^{\perp_{\mathrm{H}}}$), then $\mathcal{C}$ is called {\it Euclidean self-orthogonal}
(resp. {\it Hermitian self-orthogonal}).
Conversely, if $\mathcal{C}^{\perp_{\mathrm{E}}}\subseteq\mathcal{C}$ (resp. $\mathcal{C}^{\perp_{\mathrm{H}}}\subseteq\mathcal{C}$), then $\mathcal{C}$ is called
{\it Euclidean dual-containing} (resp. {\it Hermitian dual-containing}).

\subsection{Optimal $(r,\delta)$-LRCs}\label{LRC}

For an $[n, k, d]_q$ linear code $\mathcal{C}$, the $i$th symbol $c_i$, where $i\in[n]$,
of $\mathcal{C}$ is said to have {\it $(r, \delta)$-locality} if there exists a subset
$S_i$ $\subseteq [n]$ containing $i$ and a punctured code
$\mathcal{C}|_{S_i}$ such that the length $|S_i| \leq r + \delta -1$ and the minimum distance $d(\mathcal{C}|_{S_i})\geq \delta $, where $\mathcal{C}|_{S_i}$ is the code $\mathcal{C}$ punctured on the coordinate set $[n]\backslash S_i$ by deleting the components indexed by the set $[n]\backslash S_i$ in each codeword of $\mathcal{C}$. The code $\mathcal{C}$ is called an {\it $(r, \delta)$-LRC} if every symbol has $(r,\delta)$-locality.

For precise technical discourse, we formalize the critical notion of {\it local protection codes} through the following definition.
\begin{definition}\emph{(Local Protection Code)}\label{prote}
Let
$\mathcal{C}$ be an $(r, \delta)$-LRC of length $n$. For any coordinate position $i\in [n]$, consider a subset $S_i\subseteq [n]$ containing $i$. Then, the punctured code $\mathcal{C}|_{S_i}$ is called \emph{a local protection code} for coordinate $i$ if it satisfies:
\begin{itemize}
\item[(1)]  Locality constraint: $|S_i| \leq r + \delta -1$;
\item[(2)]  Distance guarantee: $d(\mathcal{C}|_{S_i})\geq \delta $.
\end{itemize}
\end{definition}

Let us give an example on local protection codes.

\begin{example}\label{pro-code}
Let $\mathbb{F}=\mathbb{F}_{q}$. Consider the linear code $\mathcal{C}$ with a generator matrix as follows:
\begin{align*}
G=
\begin{pmatrix}
1&0&1&0&0\\
0&1&1&0&1\\
0&0&0&1&1
\end{pmatrix}.
\end{align*}
Let $\mathbf{c}_i\in \mathbb{F}^3$ denote the $i$th column of $G$. Then, $\mathcal{C}$ is an $(2,2)$-LRC with the following local protection codes:
\begin{itemize}
\item  For symbols $c_1,c_2,c_3$, one can choose the local protection code $\mathcal{C}|_{\{1,2,3\}}=\{\mathbf{c}_1,\mathbf{c}_2,\mathbf{c}_1+\mathbf{c}_2\}$;
\item For symbols $c_4,c_5$, one can choose the local protection code $\mathcal{C}|_{\{2,4,5\}}=\{\mathbf{c}_2,\mathbf{c}_4,\mathbf{c}_2+\mathbf{c}_4\}$.
    \end{itemize}
\end{example}

\vspace{8pt}

Let $\mathcal{C}$ be an $[n,k,d]_q$ linear code with a generator matrix
$G$, expressed in column vector form as $G=(\mathbf{c}_1,\ldots,\mathbf{c}_n)$. Recall the following set:
$$C=\{(1,\mathbf{c}_1),\ldots,(n,\mathbf{c}_n)\}\subseteq [n]\times \mathbb{F}_{q}^k,$$
which is used to determine the generator matrix
$G$. For any coordinate subset $S=\{j_1,\ldots,j_i\}\subseteq [n]$, the punctured code $\mathcal{C}|_S$ is actually defined by restricting $G$ to columns indexed by $S$, yielding the submatrix $G_{S}=(\mathbf{c}_{j_1},\ldots,\mathbf{c}_{j_i})$. This punctured code is equivalently represented by the indexed subset:
\begin{align}
C_{S}=\{(j_1,\mathbf{c}_1),\ldots,(j_i,\mathbf{c}_{j_i})\}\subseteq [n]\times \mathbb{F}_{q}^k, \label{puncture}
\end{align}
uniquely determining $G_{S}$ through canonical column ordering. Let us provide the following example to illustrate these notations.

\begin{example}\label{ex:puncture}
Let $\mathcal{C} \subseteq \mathbb{F}_q^4$ be the linear code with a generator matrix as follows:
\[
G =(\mathbf{c}_1,\mathbf{c}_2, \mathbf{c}_3, \mathbf{c}_4)= \begin{pmatrix}
1 & 0 & 1 & 0 \\
0 & 1 & 1 & 1
\end{pmatrix},
\]
where $\mathbf{c}_i \in \mathbb{F}_q^2$ denotes the $i$th column of $G$. The canonically set associated with the code $\mathcal{C}$ is given by
\[
C \:= \{(1,\mathbf{c}_1), (2,\mathbf{c}_2), (3,\mathbf{c}_3), (4,\mathbf{c}_4)\} \subseteq [4] \times \mathbb{F}_q^2,
\]
enabling determination of $G$ through column ordering.

For the coordinate subset $S = \{1,2,4\}$, the \textit{punctured code} $\mathcal{C}|_S$ has the following associated set:
\[
C_S = \{(1,\mathbf{c}_1), (2,\mathbf{c}_2), (4,\mathbf{c}_4)\},
\]
which generates the submatrix as a generator matrix:
\begin{equation}\label{eq:puncture}
G_S =(\mathbf{c}_1, \mathbf{c}_2\ ,\mathbf{c}_4)= \begin{pmatrix}
1 & 0 & 0 \\
0 & 1 & 1
\end{pmatrix}.
\end{equation}
More generally, for any subset $T = \{(j_1,\mathbf{c}_{j_1}), \ldots, (j_t,\mathbf{c}_{j_t})\} \subseteq C$, the corresponding punctured code $\mathcal{C}|_T$ has a generator matrix
$$G_T = (\mathbf{c}_{j_1},\ldots,\mathbf{c}_{j_t}).$$
We define $\rank(T) \:= \dim (\langle \mathbf{c}_{j_1}, \ldots, \mathbf{c}_{j_t} \rangle)$ with $\rank(\emptyset) = 0$ by convention. These notations will be used throughout subsequent analysis.
\end{example}

We recall a basic result concerning distance characterization in linear codes.
\begin{lemma}\label{pun}
Let $\mathcal{C}$ be a linear code with length $n$, and let $C$ be the canonical set given by Eq. \eqref{cano}. Then, $\mathcal{C}$ has minimum distance $d$ if and only if the following two conditions are satisfied:
\begin{enumerate}[label=(\roman*)]
    \item There exists a subset $S_0 \subseteq C$ with cardinality $|S_0| = n-d$ satisfying $\rank(S_0) \leq k-1$;
    \item Every subset $S \subseteq C$ with $|S|> n -d$ satisfies $\rank(S) = k$.
\end{enumerate}
\end{lemma}
This fundamental characterization appears extensively in the literature (see, e.g., \cite{Mac1977,Gopalan2012}).

The {\it generalized Singleton-type bound} for $(r,\delta)$-LRCs, as established in \cite{Prakash2012}, imposes the following fundamental constraint on code parameters:
$$d\leq n-k+1-(\lceil \frac{k}{r}\rceil-1)(\delta-1),$$
where $\lceil \frac{k}{r}\rceil$ denotes the ceiling function. Codes attaining equality in this bound are called {\it optimal $(r, \delta)$-LRCs}.

We now revisit a useful lemma established in \cite{Luo2022}. Suppose that $H = (\mathbf{h}_1, \mathbf{h}_2, \ldots , \mathbf{h}_n)$ is an $m\times n$ matrix, where $\mathbf{h}_i$ is a column vector of length $m$. The {\it support} of $H$ is defined as
$$\text{Supp}(H) :=\{i \in [n]:\; \mathbf{h}_i \neq \mathbf{0}\},$$
where $\mathbf{0}$ denotes the zero column vector of length $m$.

\begin{lemma}(\!\!\cite[Lemma 2]{Luo2022}) \label{local}
Let $\mathcal{C}$ be an $[n, k, d]_q$ linear code. Let $r$ and $\delta$ be positive integers, where $\delta > 1$. The $\alpha$th code symbol of $\mathcal{C}$ has $(r,\delta)$-locality if and only if there exists an $m\times n$ matrix $H$ over $\mathbb{F}_q$ with the following properties:
\begin{itemize}
\item [(1)] $\alpha\in\text{Supp}(H)$;

\item [(2)] $|\text{Supp}(H)|\leq  r +\delta-1$ such that any $\delta-1$ nonzero
columns of $H$ are linearly independent over $\mathbb{F}_q$;

\item [(3)] $H\mathbf{c}^{\top} = \mathbf{0}$ for every codeword $\mathbf{c}\in \mathcal{C}$.
\end{itemize}
\end{lemma}

\subsection{Quantum codes}
Denote by $[\mspace{-2mu}[n,k,d]\mspace{-2mu}]_{q}$ a $q$-ary quantum code of length $n$, dimension $q^{k}$, and minimum distance $d$. Such a code is a $q^{k}$-dimensional subspace of the $q^{n}$-dimensional complex Hilbert space $(\mathbb{C}^{q})^{\otimes n} \cong \mathbb{C}^{q^{n}}$. It is capable of detecting up to $d-1$ quantum errors and correcting up to $\left\lfloor \frac{d-1}{2} \right\rfloor$ quantum errors. The minimum distance $d$ of an $[\mspace{-2mu}[n,k,d]\mspace{-2mu}]_{q}$ quantum code is constrained by the \emph{quantum Singleton bound}, which states that $2d \leq n-k+2$. A quantum code meeting this bound with equality, i.e., $2d = n - k + 2$, is referred to as a \emph{quantum MDS code}.
For further background on quantum codes, we refer the reader to \cite{Ashikhmin2001Nonbinary,Calderbank1998Quantum,Ketkar2006Nonbinary,Knill1997Theory,Gottesman1997Stabilizer,Grassl2004On,Rains1999Nonbinary}.

The CSS construction, named after Calderbank and Shor \cite{Calderbank1996Good} and Steane \cite{Steane1996Simple}, provides an efficient method for obtaining quantum codes from a pair of linear codes.

\begin{proposition}{\rm (\!\!\cite{Calderbank1996Good,Steane1996Simple}, CSS Construction)} \label{css}
Let $\mathcal{C}_{i}$ be an $[n,k_{i},d_{i}]_{q}$ linear code for $i=1,2$ such that $\mathcal{C}_{2}^{\bot_{\mathrm{E}}}\subseteq\mathcal{C}_{1}$.
Then, there exists an $[\mspace{-2mu}[n,k_{1}+k_{2}-n,d]\mspace{-2mu}]_{q}$ quantum code, where
\begin{align}\label{distanceee}
d&=\begin{cases}
\mathrm{min}\{d_1,d_2\}, & \mathrm{if}\ \mathcal{C}_{2}^{\bot_{\mathrm{E}}}=\mathcal{C}_{1},\\
\mathrm{min}\{\mathrm{wt}(\mathcal{C}_{1}\backslash\mathcal{C}_{2}^{\bot_{\mathrm{E}}}),\mathrm{wt}(\mathcal{C}_{2}\backslash\mathcal{C}_{1}^{\bot_{\mathrm{E}}})\}, & \mathrm{if}\ \mathcal{C}_{2}^{\bot_{\mathrm{E}}}\neq\mathcal{C}_{1}.
\end{cases}
\end{align}
\end{proposition}

The following proposition, known as the Hermitian construction, offers an effective method for producing a quantum code from a Hermitian dual-containing code.

\begin{proposition}{\rm (\!\!\cite{Ashikhmin2001Nonbinary}, Hermitian construction)}\label{proposition6}
Let $\mathcal{C}$ be an $[n,k,d]_{q^{2}}$ Hermitian dual-containing code. Then, there exists an $[\mspace{-2mu}[n,2k-n,\geq d]\mspace{-2mu}]_{q}$ quantum code.
\end{proposition}

\subsection{Optimal Quantum $(r, \delta)$-LRCs}
In the pioneering work of Galindo et al.~\cite{Galindo2024}, the authors established a connection between classical codes with the Hermitian dual-containing
(resp. Euclidean dual-containing) property and quantum $(r, \delta)$-LRCs. Building upon this result, they further investigated optimal quantum $(r, \delta)$-LRCs constructed from classical codes and introduced the following definition.

\begin{definition}{\rm (\!\!\cite[Definition 31]{Galindo2024})}\label{quan-LRC}
Let $\mathcal{C}\subseteq \mathbb{F}_{q^2}^n$ (resp. $\mathcal{C}\subseteq \mathbb{F}_{q}^n$) be a linear code. Suppose $\mathcal{C}$ is Hermitian dual-containing (resp. Euclidean dual-containing), $\dim(\mathcal{C})=\frac{n+k}{2}$ and $\mathcal{C}$ is an $(r, \delta)$-LRC. If $\delta \leq  d(\mathcal{C}^{\perp_{\text{H}}})$ (resp. $\delta \leq  d(\mathcal{C}^{\perp_\text{E}})$), and if the induced quantum code with parameters $[\mspace{-2mu}[n,k,\geq d(\mathcal{C})]\mspace{-2mu}]_{q}$ satisfies
\begin{align}
k + 2d(\mathcal{C}) + 2(\lceil \frac{n+k}{2r}\rceil -1)(\delta-1)=n+2,\label{quan-inqu}
\end{align}
then the induced quantum code with parameters $[\mspace{-2mu}[n,k, d(\mathcal{C})]\mspace{-2mu}]_{q}$ is said to be an {\it optimal quantum $(r, \delta)$-LRC}.
\end{definition}

To produce optimal quantum $(r, \delta)$-LRCs, we give the following lemma.

\begin{lemma}\label{quanum-cl}
Let $\mathcal{C}\subseteq \mathbb{F}_{q^2}^n$ (resp.  $\mathcal{C}\subseteq \mathbb{F}_{q}^n$) be an optimal $(r,\delta)$-LRC with parameters $[n,k,d]_{q^2}$ (resp. $[n,k,d]_{q}$). If $\mathcal{C}$ satisfies:
\vspace{-4pt}
\begin{itemize}
\item[(1)] Hermitian dual-containing (resp. Euclidean dual-containing) property;
\item[(2)] $n-k\geq \lceil \frac{k}{r}\rceil(\delta-1)$,
\end{itemize}
then its induced quantum code is an optimal quantum $(r, \delta)$-LRC.
\end{lemma}

\begin{proof}
For the Hermitian dual-containing case, namely $\mathcal{C}^{\perp_{\mathrm{H}}}\subseteq\mathcal{C}$, we have
$$d(\mathcal{C}^{\perp_\text{H}})\geq d(\mathcal{C})=n-k+1-(\lceil \frac{k}{r}\rceil-1)(\delta-1).$$
Since $n-k\geq \lceil \frac{k}{r}\rceil(\delta-1)$, we have $d(\mathcal{C}^{\perp_\text{H}})\geq\delta$.
Since $\mathcal{C}$ is an optimal $(r,\delta)$-LRC, we can verify that $2k-n+ 2d(\mathcal{C}) + 2(\lceil \frac{k}{r}\rceil -1)(\delta-1)=n+2$.
By Definition \ref{quan-LRC}, the induced quantum code is an optimal quantum $(r, \delta)$-LRC. The Euclidean case follows through an analogous proof. 
$\hfill\square$
\end{proof}

\begin{remark}
Notably, we shall rigorously demonstrate in Section \ref{quantum} that the critical inequality
$$n-k\geq \lceil \frac{k}{r}\rceil(\delta-1)$$
in Lemma \ref{quanum-cl} necessarily holds for every optimal $(r,\delta)$-LRC with parameters $[n,k,d]_{q^2}$ (resp. $[n,k,d]_{q}$) that is Hermitian dual-containing (resp. Euclidean dual-containing). This finding fundamentally strengthens the applicability of Lemma \ref{quanum-cl}.
\end{remark}

\section{The decomposition theorem for optimal $(r,\delta)$-LRCs}\label{OD}
In this section, we will build the decomposition theorem for optimal $(r,\delta)$-LRCs.

\subsection{Decomposition Theorem}\label{sub3.1}

Let $\mathcal{C}$ be an $[n, k, d]_q$ linear code with a generator matrix $G=(\mathbf{c}_1,\ldots,\mathbf{c}_n)$. Recall the following canonical set:
\begin{align}
C=\{(1,\mathbf{c}_1),\ldots,(n,\mathbf{c}_n)\}\subseteq [n]\times \mathbb{F}_{q}^k. \label{setC}
\end{align}
For any subset $T=\{(j_1,\mathbf{c}_{j_1}),\ldots,(j_i,\mathbf{c}_{j_i})\}\subseteq C$, recall that the rank of $T$ is defined as
$$\rank(T):=\dim(\langle \mathbf{c}_{j_1},\ldots,\mathbf{c}_{j_i}\rangle),$$
with the convention $\rank(\emptyset)=0$.

Next, we present three useful lemmas, namely Lemmas \ref{Stru-lemm1}, \ref{Stru-lemm2} and \ref{Stru-lemm3}, which will aid in establishing the decomposition theorem for optimal $(r,\delta)$-LRCs.

\begin{lemma}\label{Stru-lemm1}
Let $\mathcal{C}$ be an optimal $(r,\delta)$-LRC with parameters $[n,k,d]_q$, and let $C$ be the set defined in Eq. \eqref{setC}. Suppose there exists a subset $S\subseteq C$ such that
$$S=C_1\cup\ldots\cup C_{i_0}\cup \{(j_1,\mathbf{c}_{j_1}),\ldots,(j_{s_{i_0+1}},\mathbf{c}_{j_{s_{i_0+1}}})\},$$ where $\{j_1,\ldots,j_{s_{i_0+1}}\}\subseteq [n]$ and $C_i\subseteq C$ for $i\in[i_{0}]$, satisfying the following conditions:
\begin{itemize}
\item [(1)] Rank Constraint: $\rank(S)=k-1$ and $s_{i_0+1}\leq r-1$;

\item [(2)] Incremental Rank Growth: For all $i\in [i_0]$:
$$1\leq \rank(\cup_{j=1}^i C_j)-\rank(\cup_{j=1}^{i-1} C_j)\leq (|\cup_{j=1}^i C_j|-|\cup_{j=1}^{i-1} C_j|-(\delta-1)),$$
where $C_0=\emptyset$;

\item [(3)] Local Protection: Each $\mathcal{C}|_{C_i}$ is a local protection code for $i\in [i_0]$.
    \end{itemize}
Then, the following statements hold:
\begin{itemize}
\item  $i_0=\lceil \frac{k}{r}\rceil-1$;

\item $\mathcal{C}|_{C_1}$ is an MDS code with parameters $[n_1\leq (r+\delta-1),n_1-\delta+1,\delta]_{q}$;

\item $d(\mathcal{C})=n-|S|$.
    \end{itemize}
\end{lemma}

\begin{proof}
Case 1: If $\lceil\frac{k}{r}\rceil=1$, then
$\mathcal{C}$ is an MDS code, and the conclusion holds trivially.

Case 2: Assume $\lceil\frac{k}{r}\rceil>1$. Let $S\subseteq C$ such that:
$$S=C_1\cup\ldots\cup C_{i_0}\cup \{(j_1,\mathbf{c}_{j_1}),\ldots,(j_{s_{i_0+1}},\mathbf{c}_{j_{s_{i_0+1}}})\},$$
where $\{j_1,\ldots,j_{s_{i_0+1}}\}\subseteq [n]$ and $C_i\subseteq C$ for $i\in[i_{0}]$ satisfy the above conditions (1)-(3). Define the following numbers:
\begin{align}
s_i:=|\cup_{j=1}^i C_j|-|\cup_{j=1}^{i-1} C_j|,\;r_i:=\rank(\cup_{j=1}^i C_j)-\rank(\cup_{j=1}^{i-1} C_j), \label{count2}
\end{align}
where $i\in[i_{0}]$ and $C_0:=\emptyset$. Since $\mathcal{C}|_{C_i}$ is a local protection code for $i\in [i_0]$, we have
$$|C_i|\leq r+\delta -1.$$
This implies that
$$s_i\leq r+\delta -1.$$
Combining this with the definition of $r_i$, we have
\begin{numcases}{}
	 s_i \leq r+\delta -1 \quad (i\in[i_{0}]), \label{ineq1} \\
	r_i \leq s_i-(\delta-1) \quad (i\in[i_{0}]). \label{ineq22}
\end{numcases}
Let $r_{i_0+1}:=\rank(S)-\rank(\cup_{j=1}^{i_0} C_j)$. Since $s_{i_0+1}\leq r-1$ and by the definition of $r_{i_0+1}$, we obtain
\begin{numcases}{}
	 s_{i_0+1}\leq r-1, \label{ineq3} \\
	r_{i_0+1}\leq s_{i_0+1}. \label{ineq4}
\end{numcases}
By condition (1), we know that
\begin{align}
	 \rank(S)=k-1. \label{rS}
\end{align}
By Lemma \ref{pun}, $d(\mathcal{C})\leq n-|S|$. Note that $|S|=\sum_{i=1}^{i_0+1}s_i$, which implies that
\begin{align}
	 d(\mathcal{C}) \leq n-\sum_{i=1}^{i_0+1}s_i. \label{ineq5}
\end{align}
Since $\mathcal{C}$ is an optimal $(r,\delta)$-LRC, $d(\mathcal{C})=n-k+1-(\lceil \frac{k}{r}\rceil-1)(\delta-1)$. Combining this with Eq. \eqref{ineq5}, we obtain
\begin{align}
	 \sum_{i=1}^{i_0+1}s_i\leq  k-1+(\lceil \frac{k}{r}\rceil-1)(\delta-1).\label{ineq6}
\end{align}
By Eqs. \eqref{ineq22} and \eqref{ineq4}, $\sum_{i=1}^{i_0+1}r_i +i_0(\delta-1)\leq \sum_{i=1}^{i_0+1}s_i$. Note that $\sum_{i=1}^{i_0+1}r_i=\rank(S)$. Combining this with Eq. \eqref{rS}, we obtain $\sum_{i=1}^{i_0+1}r_i=k-1$. Since we have already established that $\sum_{i=1}^{i_0+1}r_i +i_0(\delta-1)\leq \sum_{i=1}^{i_0+1}s_i$, the following inequality holds:
\begin{align}
	 k-1 +i_0(\delta-1)\leq \sum_{i=1}^{i_0+1}s_i. \label{ineq7}
\end{align}
Combining this with Eq. \eqref{ineq6}, we establish
\begin{align}
	 k-1 +i_0(\delta-1) \leq k-1+(\lceil \frac{k}{r}\rceil-1)(\delta-1). \label{ineq8}
\end{align}
This implies
\begin{align}
	 i_0\leq \lceil \frac{k}{r}\rceil-1. \label{ineq9}
\end{align}
On the other hand, due to Eqs. \eqref{ineq1} and \eqref{ineq3}, we obtain $\sum_{i=1}^{i_0+1}s_i\leq i_0(r+\delta-1)+r-1$. Combining this with Eq. \eqref{ineq7}, we conclude $k-1 +i_0(\delta-1)\leq i_0(r+\delta-1)+r-1$. This implies
\begin{align}
	 k \leq (i_0+1)r. \label{ineq10}
\end{align}
From this, we obtain
\begin{align}
	 i_0\geq \lceil \frac{k}{r}\rceil-1. \label{ineq11}
\end{align}
Combining Eqs. \eqref{ineq9} and \eqref{ineq11}, we obtain $i_0=\lceil \frac{k}{r}\rceil-1$.

Next, we show that $\mathcal{C}|_{C_1}$ is an MDS code with parameters $[n_1\leq (r+\delta-1),k_1,\delta]_{q}$. Since $i_0=\lceil \frac{k}{r}\rceil-1$ and Eq. \eqref{ineq7}, we have
\begin{align}
\sum_{i=1}^{i_0+1}s_i \geq  k-1+(\lceil \frac{k}{r}\rceil-1)(\delta-1).\label{ineq12}
\end{align}
By Eqs. \eqref{ineq12} and \eqref{ineq6}, the following equation holds:
\begin{align}
\sum_{i=1}^{i_0+1}s_i=k-1+(\lceil \frac{k}{r}\rceil-1)(\delta-1).\label{ineq13}
\end{align}
Note that $k-1=\sum_{i=1}^{i_0+1}r_i$, hence Eq. \eqref{ineq13} becomes
\begin{align}
\sum_{i=1}^{i_0+1}s_i=\sum_{i=1}^{i_0+1}r_i+(\lceil \frac{k}{r}\rceil-1)(\delta-1).\label{ineq14}
\end{align}
Since $i_0=\lceil \frac{k}{r}\rceil-1$, and by Eqs. \eqref{ineq22} and \eqref{ineq4}, the above equality holds only if
$$r_i=s_i-(\delta-1)\;,\;r_{i_0+1}=s_{i_0+1},$$
where $i\in[i_{0}]$. In particular, $r_1=s_1-(\delta-1)$. Let $n_1$ be the length of $\mathcal{C}|_{C_1}$, and let $k_1$ be the dimension of $\mathcal{C}|_{C_1}$. By definition, $n_1=s_1$ and $k_1=r_1$. Since $r_1=s_1-(\delta-1)$, we know $k_1=n_1-(\delta-1)$. Because $\mathcal{C}|_{C_1}$ is a local protection code, $d(\mathcal{C}|_{C_1})\geq \delta$. On the other hand, $d(\mathcal{C}|_{C_1})\leq n_1-k_1+1=n_1-(n_1-(\delta-1))+1=\delta$ due to the Singleton bound, hence $d(\mathcal{C}_1)=\delta$. This implies that $\mathcal{C}_1$ is an MDS code with parameters $[n_1\leq (r+\delta-1),k_1,\delta]_{q}$.
Since $|S|=\sum_{i=1}^{i_0+1} s_i=k-1+(\frac{k}{r}\rceil-1)(\delta-1)$ by Eq. \eqref{ineq13}, it follows that $n-|S|=d(\mathcal{C})$.

Therefore, we complete the proof. $\hfill\square$
\end{proof}

\vspace{8pt}

Next, we proceed to address the analysis of another decomposition framework.
\begin{lemma}\label{Stru-lemm2}
Let $\mathcal{C}$ be an optimal $(r,\delta)$-LRC with parameters $[n,k,d]_q$, and let $C$ be the set defined in Eq. (\ref{setC}). Suppose there exists a subset $S\subseteq C$ such that $S=\mathcal{R}_1\cup \ldots \cup \mathcal{R}_{j_0}$ for some $j_0 \in \mathbb{N}$ and $\mathcal{R}_i\subseteq C$ for $i\in[j_{0}]$, satisfying
the following conditions:
\begin{itemize}
\item [(1)] Rank Constraint: $\rank(S)=k-1$;

\item [(2)] Incremental Rank Growth:  for all $i\in [j_0]$,
$$1\leq \rank(\cup_{j=1}^i \mathcal{R}_j)-\rank(\cup_{j=1}^{i-1} \mathcal{R}_j)\leq (|\cup_{j=1}^i \mathcal{R}_j|-|\cup_{j=1}^{i-1} \mathcal{R}_j|-(\delta-1)),$$
where $\mathcal{R}_0:=\emptyset$;

\item [(3)] Local Protection: Each $\mathcal{C}|_{\mathcal{R}_i}$ is a local protection code for $i\in [j_0]$.
    \end{itemize}
Then, the following statements hold:
\begin{itemize}
\item  $j_0=\lceil \frac{k}{r}\rceil-1$;

\item $\mathcal{C}|_{\mathcal{R}_1}$ is an MDS code with parameters $[n_1\leq (r+\delta-1),n_1-\delta+1,\delta]_{q}$;

\item $d(\mathcal{C})=n-|S|$.
    \end{itemize}
\end{lemma}

\begin{proof}
Similar to the proof of Lemma \ref{Stru-lemm1}, define
\begin{align}
s_i':=|\cup_{j=1}^i\mathcal{R}_{i}|-|\cup_{j=1}^{i-1}\mathcal{R}_{i}|,\;r_i':=\rank(\cup_{j=1}^i\mathcal{R}_{i})-\rank(\cup_{j=1}^{i-1}\mathcal{R}_{i}),\label{count1}
\end{align}
where $i\in [j_{0}]$ and $\mathcal{R}_0=\emptyset$. Since $\mathcal{C}|_{\mathcal{R}_i}$ is a local protection code for $i\in [j_{0}]$ and by using the condition (2), we obtain
\begin{numcases}{}
	 s_i' \leq r+\delta -1 \quad (i\in [j_{0}]), \label{iineq1} \\
	r_i' \leq s_i'-(\delta-1) \quad (i\in [j_{0}]). \label{iineq2}
\end{numcases}
Summing these inequalities yields
\begin{numcases}{}
	 \sum_{i=1}^{j_0} s_i' \leq j_0(r+\delta -1), \label{iineq3} \\
	\sum_{i=1}^{j_0} r_i' \leq \sum_{i=1}^{j_0} s_i'-j_0(\delta -1).  \label{iineq4}
\end{numcases}
Since $\rank(S)=k-1$ and $\rank(S)=\sum_{i=1}^{j_0} r_i'$ by definition, we obtain $\sum_{i=1}^{j_0} r_i'=k-1$. Combining this with Eq. \eqref{iineq4}, we derive
\begin{align}
k-1\leq \sum_{i=1}^{j_0} s_i'-j_0(\delta -1).\label{iineq5}
\end{align}
Combining this with Eq. \eqref{iineq3}, we have $k-1\leq j_0(r+\delta -1)-j_0(\delta -1)$. From this, we obtain
\begin{align}
j_0+1\geq \lceil \frac{k}{r}\rceil.\label{iineq6}
\end{align}
Since $\rank(S)=k-1$, we have $d(\mathcal{C})\leq n-\sum_{i=1}^{j_0} s_i'$ by Lemma \ref{pun}. Note that $\sum_{i=1}^{j_0} s_i'\geq (k-1)+j_0(\delta-1)$ by Eq. \eqref{iineq5}, thus
\begin{align}
d(\mathcal{C})\leq n-k+1-j_0(\delta-1).\label{iineq7}
\end{align}
Since $\mathcal{C}$ is an optimal $(r,\delta)$-LRC, we have $d(\mathcal{C})=n-k+1-(\lceil \frac{k}{r}\rceil-1)(\delta-1)$. Combining this with Eq. \eqref{iineq7}, the following inequality holds:
\begin{align}
	j_0+1\leq \lceil \frac{k}{r}\rceil. \label{iineq8}
\end{align}
Using Eqs. \eqref{iineq6} and \eqref{iineq8}, we conclude
\begin{align}
	j_0=\lceil \frac{k}{r}\rceil-1. \label{iineq9}
\end{align}

Then, we show that $\mathcal{C}|_{\mathcal{R}_1}$ is an MDS code with parameters $[n_1\leq (r+\delta-1),k_1,\delta]_{q}$. Since $\mathcal{C}$ is an optimal $(r,\delta)$-LRC and $d(\mathcal{C})\leq n-\sum_{i=1}^{j_0} s_i'$, we obtain $n-k+1-(\lceil \frac{k}{r}\rceil-1)(\delta-1)\leq n-\sum_{i=1}^{j_0} s_i'$. Hence,
\begin{align}
	\sum_{i=1}^{j_0} s_i'\leq k-1+(\lceil \frac{k}{r}\rceil-1)(\delta-1).\label{iineq10}
\end{align}
Using Eq. \eqref{iineq9}, the above inequality becomes
\begin{align}
	\sum_{i=1}^{j_0} s_i'\leq k-1+j_0(\delta-1).\label{iineq11}
\end{align}
Now, applying Eqs. \eqref{iineq11} and \eqref{iineq5}, we have
\begin{align}
	\sum_{i=1}^{j_0} s_i'=k-1+j_0(\delta-1).\label{iineq12}
\end{align}
Note that $k-1=\sum_{i=1}^{j_0} r_i'$, thus
\begin{align}
	\sum_{i=1}^{j_0} s_i'=\sum_{i=1}^{j_0} r_i'+j_0(\delta-1).\label{iineq13}
\end{align}
By Eq. \eqref{iineq2}, Eq. \eqref{iineq13} holds only if
\begin{align}
	r_i'=s_i'-(\delta-1)
\end{align}
for $i\in[j_{0}]$. In particular, $r_1'=s_1'-(\delta-1)$. Let $n_1$ be the length of $\mathcal{C}|_{\mathcal{R}_1}$, and let $k_1$ be the dimension of $\mathcal{C}|_{\mathcal{R}_1}$. By definition, $n_1=s_1'$ and $k_1=r_1'$. Since we have already established that $r_1'=s_1'-(\delta-1)$, we obtain $k_1=n_1-(\delta-1)$. Since $\mathcal{C}|_{\mathcal{R}_1}$ is a local protection code, we obtain $d(\mathcal{C}|_{\mathcal{R}_1})\geq \delta$. On the other hand, $d(\mathcal{C}|_{\mathcal{R}_1})\leq n_1-k_1+1\leq n_1-(n_1-(\delta-1))+1=\delta$ due to the Singleton bound, hence $d(\mathcal{C}|_{\mathcal{R}_1})=\delta$. From this, we derive that $\mathcal{C}|_{\mathcal{R}_1}$ is an MDS code with parameters $[n_1\leq (r+\delta-1),k_1,\delta]_{q}$.

By Eq. \eqref{iineq12}, $|S|=k-1+j_0(\delta-1)$. Since $j_0+1=\lceil \frac{k}{r}\rceil$ by Eq. \eqref{iineq9}, we obtain
\begin{align}
	|S|=k-1+(\lceil \frac{k}{r}\rceil-1)(\delta-1).
\end{align}
Hence, $n-|S|=n-k+1-(\lceil \frac{k}{r}\rceil-1)(\delta-1)=d(\mathcal{C})$. Therefore, we complete the proof. $\hfill\square$
\end{proof}

\vspace{8pt}

We now establish that any optimal $(r,\delta)$-LRC necessarily admits a decomposition structure as specified in either Lemma \ref{Stru-lemm1} or Lemma \ref{Stru-lemm2}.
\begin{lemma}\label{Stru-lemm3}
Let $\mathcal{C}$ be an optimal $(r,\delta)$-LRC with parameters $[n,k,d]_q$, and let $C$ be the set defined in Eq. (\ref{setC}). Then, $C$ must have one of the following two decompositions:\\
\underline{Case (I)}:
$$C=C_1\cup\ldots\cup C_{t-1}\cup \{(j_1,\mathbf{c}_{j_1}),\ldots,(j_{s_{t}},\mathbf{c}_{j_{s_t}})\}\cup U,$$
where $t=\lceil \frac{k}{r}\rceil$, $\{j_1,\ldots,j_{s_t}\}\subseteq [n]$, $C_i\subseteq C$ for $i\in[t-1]$, $\{j_1,\ldots,j_{s_t}\}$ is contained in some $T\subseteq C$ such that $\mathcal{C}|_{T}$ is a local protection code, and
\begin{itemize}
\item [(1)] Local Protection: Each $\mathcal{C}|_{C_i}$ is a local protection code for $i\in[t-1]$, and it is an MDS code with parameters $[n_i\leq (r+\delta-1),k_i,\delta]_{q}$;
\item [(2)] Rank-Size Equations: For all $i\in[t-1]$,
$$1\leq \rank(\cup_{j=1}^i C_j)-\rank(\cup_{j=1}^{i-1} C_j)=|\cup_{j=1}^i C_j|-|\cup_{j=1}^{i-1} C_j|-(\delta-1),$$
where $C_0:=\emptyset$;
\item [(3)] Terminal Rank:
$$\rank(\cup_{j=1}^{t-1} C_j\cup \{(j_1,\mathbf{c}_{j_1}),\ldots,(j_{s_{t}},\mathbf{c}_{j_{s_t}})\})=\rank(\cup_{j=1}^{t-1} C_j)+s_t=k-1;$$
\item [(4)] Disjoint Residual Set: $U\cap(C_1\cup\ldots\cup C_{t-1}\cup \{(j_1,\mathbf{c}_{j_1}),\ldots,(j_{s_{t}},\mathbf{c}_{j_{s_t}})\})=\emptyset$ and $|U|=d(\mathcal{C})$.
    \end{itemize}

\underline{Case (II)}:
$$C=C_1\cup\ldots\cup C_{t} \cup U,$$
where $t=\lceil \frac{k}{r}\rceil -1$, $C_i\subseteq C$ for $i\in[t]$, and
\begin{itemize}
\item [(1)] Local Protection: Each $\mathcal{C}|_{C_i}$ is a protected code for $i\in [t]$, and it is an MDS code with parameters $[n_i\leq (r+\delta-1),k_i,\delta]_{q}$;
\item [(2)] Rank-Size Equations: For all $i\in [t]$,
$$1\leq \rank(\cup_{j=1}^i C_j)-\rank(\cup_{j=1}^{i-1} C_j)=|\cup_{j=1}^i C_j|-|\cup_{j=1}^{i-1} C_j|-(\delta-1),$$
where $C_0:=\emptyset$;
\item [(3)] Terminal Rank: $\rank(C_1\cup\ldots\cup C_{t})=k-1$;
\item [(4)] Disjoint Residual Set: $U\cap (\cup_{j=1}^{t} C_j)=\emptyset$ and $|U|=d(\mathcal{C})$.
\end{itemize}
\end{lemma}

\begin{proof}
Let $G=(\mathbf{c}_1,\ldots,\mathbf{c}_n)$ be a generator matrix of $\mathcal{C}$, and let $C=\{(1,\mathbf{c}_1),\ldots,(n,\mathbf{c}_n)\}$. Select any subset $\{(i_1,\mathbf{c}_{i_1}),\ldots,(i_k,\mathbf{c}_{i_k})\}\subseteq C$ with $\rank(\langle \mathbf{c}_{i_1},\ldots,\mathbf{c}_{i_k}\rangle )=k$. Execute the following iterative procedure to construct a subset $S\subseteq C$:
\begin{algorithm}[htbp]
\caption{Subset Selection for Rank Attainment} \label{alg:rank-attainment}
\begin{algorithmic}

\Statex
\State \textbf{P1.} Initialize $j \gets 1$, $S_0 \gets \emptyset$
\While{$\text{rank}(S_{j-1}) \leq k - 2$}
    \State \textbf{P2.} Pick $(i_j,\mathbf{ c}_{i_j}) \in C \setminus S_{j-1}$
    \State \textbf{P3.} Choose a local protection code $\mathcal{C}|_{S_j}$ for coordinate $i_j$
    \If{$\text{rank}(S_{j-1} \cup S_j) < k$}
        \State \textbf{P4.} $S_j \gets S_{j-1} \cup S_j$
      \Else
        \State Pick $T_j \subseteq S_j$ such that $\text{rank}(S_{j-1} \cup T_j) = k - 1$
        \State \textbf{P5.} $S_j \gets S_{j-1} \cup T_j$
    \EndIf
   \State \textbf{P6.} $j \gets j + 1$
\EndWhile
\end{algorithmic}
\end{algorithm}

Here, $\text{P}1, \ldots, \text{P}6$ are numbered for later use. Since $\rank(\langle \mathbf{c}_{i_1},\ldots,\mathbf{c}_{i_k}\rangle )=k$, the iterative process must terminate at some finite step $j$, yielding a final constructed set $S=S_j$. When the algorithm terminates via Step P5, the resulting set $S$ must admits  a decomposition:
$$S=C_1\cup \ldots\cup C_{i_0}\cup \{(j_1,\mathbf{c}_{j_1}),\ldots,(j_{s_{i_0+1}},\mathbf{c}_{j_{s_{i_0+1}}})\},$$
where $i_0,s_{i_0+1}\in \mathbb{N}$ and $\mathcal{C}|_{C_i}$ are local protection codes for $i\in[i_{0}]$.
By the termination Step P5, there exists $S_{i_0+1}\subseteq C$ satisfying:
\begin{itemize}
\item [(i)] $\mathcal{C}_{S_{i_0+1}}$ is a local protection code;
\item [(ii)] $\{(j_1,\mathbf{c}_{j_1}),\ldots,(j_{s_{i_0+1}},\mathbf{c}_{j_{s_{i_0+1}}})\}\subseteq S_{i_0+1}$;

\item [(iii)] For $T=\cup_{j=1}^{i_0} C_j\cup S_{i_0+1}$, we have $\rank(\mathcal{C}_{T})=k$.
\end{itemize}
The local protection property of $\mathcal{C}_{S_{i_0+1}}$ implies that
$$\rank(T)-\rank(\cup_{j=1}^{i_0} C_j)\leq r.$$
Consequently, there exists a subset $T''\subseteq S_{i_0+1}$ such that $|T''|\leq r$ and $\rank(\cup_{j=1}^{i_0} C_j\cup T'')=\rank(\mathcal{C}_{T})=k$.
By the construction of $S$, we may select $T''$ to contain the terminal subset, i.e., the following relation holds: $$\{(j_1,\mathbf{c}_{j_1}),\ldots,(j_{s_{i_0+1}},\mathbf{c}_{j_{s_{i_0+1}}})\}\subseteq T''.$$
The rank inequality
$$\rank(\cup_{j=1}^{i_0} C_j\cup \{(j_1,\mathbf{c}_{j_1}),\ldots,(j_{s_{i_0+1}},\mathbf{c}_{j_{s_{i_0+1}}})\})=k-1< \rank(\cup_{j=1}^{i_0} C_j\cup T'')=k$$
implies that $\{(j_1,\mathbf{c}_{j_1}),\ldots,(j_{s_{i_0+1}},\mathbf{c}_{j_{s_{i_0+1}}})\}$ is a proper subset of $T''$. Thus, $$|\{(j_1,\mathbf{c}_{j_1}),\ldots,(j_{s_{i_0+1}},\mathbf{c}_{j_{s_{i_0+1}}})\}|\leq |T''|-1\leq r-1.$$
By the definition of $s_{i_0+1}$, we have $s_{i_0+1}\leq r-1$. Next, we check that the set $S$ satisfies all the conditions of Lemma \ref{Stru-lemm1}.

The constructed set $S$ satisfies $\rank(S)=k-1$ due to the above algorithm. Since we have already known that $s_{i_0+1}\leq r-1$, the condition (1) of Lemma \ref{Stru-lemm1} holds. According to the above program's Step P3, the condition (2) of Lemma \ref{Stru-lemm1} holds. Since $\mathcal{C}|_{C_i}$ is a local protection code for $i\in[i_{0}]$, the condition (3) of Lemma \ref{Stru-lemm1} holds. Now, we can apply Lemma \ref{Stru-lemm1} to conclude:
\begin{itemize}
\item [(i)] $\mathcal{C}|_{C_1}$ is an MDS code with parameters $[n_1\leq (r+\delta-1),n_1-\delta+1,\delta]_{q}$;
\item [(ii)] The conditions (2)-(4) of this lemma hold.
\end{itemize}
Then, we post the proof that the condition (1) of this lemma holds.

When the algorithm terminates via Step P4, the resulting set $S$ must admits a decomposition:
\begin{align}
S=\mathcal{R}_1\cup \ldots \cup \mathcal{R}_{j_0}, \label{decom}
\end{align}
with some $j_0 \in \mathbb{N}$, where each $\mathcal{C}|_{\mathcal{R}_i}$ is a local protection code for $i\in [j_{0}]$, and $\rank(S)=k-1$, i.e., the condition (1) of Lemma \ref{Stru-lemm1} holds. Since Step P3 of the above program, the condition (2)-(3) of Lemma \ref{Stru-lemm1} hold. Hence, we can apply Lemma \ref{Stru-lemm1} to derive that the conditions (2)-(4) in Case (II) hold.

Next, we show that $\mathcal{C}|_{\mathcal{R}_i}$ is an MDS code with parameters $[n_i\leq (r+\delta-1),k_i,\delta]_{q}$ for $i\in[j_0]$. By Lemma \ref{Stru-lemm2},  $\mathcal{C}|_{\mathcal{R}_1}$ is an MDS code with parameters $[n_1\leq (r+\delta-1),n_1-\delta+1,\delta]_{q}$. Let $i\in [j_0]\backslash \{1\}$. Note that the choice of $S_1$ in the above program can be any subset $S_1\subseteq C$ that satisfies that $\mathcal{C}|_{S_1}$ is a local protection code. Therefore, we can assume that $\mathcal{R}_i$ corresponds to the set  $C_1$ in Case (I) or to the set $\mathcal{R}_1$ in Eq. (\ref{decom}). Since we have established that both $\mathcal{C}|_{C_1}$ and $\mathcal{C}|_{\mathcal{R}_1}$ are MDS codes with parameters of the form $[n_1\leq (r+\delta-1),k_1,\delta]_{q}$, it follows that $\mathcal{C}|_{\mathcal{R}_i}$ is also an MDS code with parameters $[n_1\leq (r+\delta-1),k_1,\delta]_{q}$, i.e., the condition (1) of this lemma holds.

Similarly, we can consider $C_i$ to be the set $C_1$ in Case (I) or the set $\mathcal{R}_1$ in Eq. (\ref{decom}). Given that we have already demonstrated that both the code $\mathcal{C}|_{C_1}$ in Case (I) and the code $\mathcal{C}|_{\mathcal{R}_1}$ are MDS codes with parameters of the form $[n_1\leq (r+\delta-1),k_1,\delta]_{q}$, it follows that $\mathcal{C}|_{C_i}$ is also an MDS code with parameters of the form $[n_i\leq (r+\delta-1),k_i,\delta]_{q}$. This fulfills the proof that the condition (1) of this lemma holds. Therefore, we complete the proof. $\hfill\square$
\end{proof}

\vspace{8pt}

Combining Lemmas \ref{Stru-lemm1}, \ref{Stru-lemm2} and \ref{Stru-lemm3}, we present the decomposition theorem for optimal $(r,\delta)$-LRCs in the following theorem.

\begin{theorem}(\emph{Decomposition Theorem})\label{Stru}
Suppose $\mathcal{C}$ is optimal $(r,\delta)$-LRC with parameters $[n,k,d]_q$. Let $C$ be the set defined in Eq. \eqref{setC}. Then, $C$ must have one of the following two decompositions:\\
\underline{Case (I)}:
$$C=C_1\cup\ldots\cup C_{t-1}\cup \{(j_1,\mathbf{c}_{j_1}),\ldots,(j_{s_{t}},\mathbf{c}_{j_{s_t}})\}\cup U,$$
where $t=\lceil \frac{k}{r}\rceil$, $\{j_1,\ldots,j_{s_t}\}\subseteq [n]$, $C_i\subseteq C$ for $i\in[t-1]$, $\{j_1,\ldots,j_{s_t}\}$ is contained in some $T\subseteq C$ such that $\mathcal{C}|_{T}$ is a local protection code, and
\begin{itemize}
\item [(1)] Local Protection: Each $\mathcal{C}|_{C_i}$ is a local protection code for $i\in [t-1]$, and it is an MDS code with parameters $[n_i\leq (r+\delta-1),k_i,\delta]_{q}$;
\item [(2)] Rank-Size Equations: For all $i\in [t-1]$,
$$1\leq \rank(\cup_{j=1}^i C_j)-\rank(\cup_{j=1}^{i-1} C_j)=|\cup_{j=1}^i C_j|-|\cup_{j=1}^{i-1} C_j|-(\delta-1),$$
where $C_0:=\emptyset$;
\item [(3)] Terminal Rank:
$$\rank(\cup_{j=1}^{t-1} C_j\cup \{(j_1,\mathbf{c}_{j_1}),\ldots,(j_{s_{t}},\mathbf{c}_{j_{s_t}})\})=\rank(\cup_{j=1}^{t-1} C_j)+s_t=k-1;$$
\item [(4)] Disjoint Residual Set: $U\cap(C_1\cup\ldots\cup C_{t-1}\cup \{(j_1,\mathbf{c}_{j_1}),\ldots,(j_{s_{t}},\mathbf{c}_{j_{s_t}})\})=\emptyset$ and $|U|=d(\mathcal{C})$.
    \end{itemize}

\underline{Case (II)}:
$$C=C_1\cup\ldots\cup C_{t} \cup U,$$
where $t=\lceil \frac{k}{r}\rceil -1$, $C_i\subseteq C$ for $i\in[t]$, and
\begin{itemize}
\item [(1)] Local Protection: Each $\mathcal{C}|_{C_i}$ is a local protection code for $i\in [t]$, and it is an MDS code with parameters $[n_i\leq (r+\delta-1),k_i,\delta]_{q}$;
\item [(2)] Rank-Size Equations: for all $i\in [t]$,
$$1\leq \rank(\cup_{j=1}^i C_j)-\rank(\cup_{j=1}^{i-1} C_j)=|\cup_{j=1}^i C_j|-|\cup_{j=1}^{i-1} C_j|-(\delta-1),$$
where $C_0:=\emptyset$;
\item [(3)] Terminal Rank: $\rank(C_1\cup\ldots\cup C_{t})=k-1$;
\item [(4)] Disjoint Residual Set: $U\cap (\cup_{j=1}^{t} C_j)=\emptyset$ and $|U|=d(\mathcal{C})$.
\end{itemize}
Conversely, assume that $\mathcal{C}$ is just a code with parameters $[n,k,d]_{q}$ that exhibits the decomposition of Case (I). If for every $(i,\mathbf{c}_{i})\in \{(j_1,\mathbf{c}_{j_1}),\ldots,(j_{s_t},\mathbf{c}_{j_{s_t}})\}\cup U$, $\mathcal{C}$ possesses $(r,\delta)$-locality for the $i$th symbol, then $\mathcal{C}$ is an optimal $(r,\delta)$-LRC. Similarly, if $\mathcal{C}$ has a decomposition of Case (II), and for every $(i,\mathbf{c}_{i})\in U$, $\mathcal{C}$ possesses $(r,\delta)$-locality for the $i$th symbol, then $\mathcal{C}$ is also an optimal $(r,\delta)$-LRC.
\end{theorem}
\begin{proof}
Assume $\mathcal{C}$ is an optimal $(r,\delta)$-LRC with parameters $[n,k,d]_{q}$. By Lemma \ref{Stru-lemm3}, $\mathcal{C}$ must have a decomposition of either Case (I) or Case (II).

Reversely, assume that $\mathcal{C}$ is just a code with parameters $[n,k,d]_{q}$ that has a decomposition of Case (I), and that for every $(i,\mathbf{c}_{i})\in \{(j_1,\mathbf{c}_{j_1}),\ldots,(j_{s_t},\mathbf{c}_{j_{s_t}})\}\cup U$, $\mathcal{C}$ possesses $(r,\delta)$-locality for the $i$th symbol. Let $i\in [t-1]$. Since $\mathcal{C}|_{C_i}$ is an MDS code with parameters $[n_i\leq (r+\delta-1),k_i,\delta]$, it follows that $\mathcal{C}$ has $(r, \delta)$-locality for the $i$th symbol when using the local protection code $\mathcal{C}|_{C_i}$, where $i\in \{j:\;(j,\mathbf{c}_j)\in C_i\}$.

Since for every $(i,\mathbf{c}_{i})\in \{(j_1,\mathbf{c}_{j_1}),\ldots,(j_{s_t},\mathbf{c}_{j_{s_t}})\}\cup U$, $\mathcal{C}$ possesses $(r,\delta)$-locality for the $i$th symbol, it follows that $\mathcal{C}$ is an $(r, \delta)$-LRC.
By the conditions (2) and (3) in Case (I), the number of elements in $C_1\cup\ldots\cup C_{t-1}\cup \{(j_1,\mathbf{c}_{j_1}),\ldots,(j_{s_t},\mathbf{c}_{j_{s_t}})\}$ is $k-1+(\lceil \frac{k}{r}\rceil-1)(\delta-1)$. Thus, the size of $U$ is given by $$|U|=n-k+1-(\lceil \frac{k}{r}\rceil-1)(\delta-1).$$
From condition (4), we have $d(\mathcal{C})=|U|=n-k+1-(\lceil \frac{k}{r}\rceil-1)(\delta-1)$, which implies that $\mathcal{C}$ is an optimal $(r,\delta)$-LRC.

Similarly, one can show that $\mathcal{C}$ is an optimal $(r,\delta)$-LRC when $\mathcal{C}$ is just a code with parameters $[n,k,d]_{q}$ that has a decomposition in Case (II), and for every $(i,\mathbf{c}_{i})\in U$, $\mathcal{C}$ possesses $(r,\delta)$-locality for the $i$th symbol.

Therefore, we complete the proof. $\hfill\square$
\end{proof}

\begin{remark}
Prior literature has established decomposition theorems for the cases where $r|k$ or some other restrictive conditions
(see, e.g., \cite[Theorem 5]{Prakash2012} and \cite[Theorem 9]{Song2014}). However, the structural decomposition for general cases, particularly the critical case of $r|(k-1)$, remained unresolved. In Proposition \ref{Stru-prop1}, we shall prove that the decomposition Case (II) in Theorem \ref{Stru} can only emerge under the condition $r|(k-1)$.
\end{remark}

We now present concrete examples to demonstrate both decomposition cases from Theorem \ref{Stru}.

\begin{example}\label{Stru-example1}
Let $\mathbb{F}=\mathbb{F}_{q}$. Consider a $(3,2)$-LRC with the following generator matrix:
\begin{align*}
G=(\mathbf{c}_1,\ldots, \mathbf{c}_6)=
\begin{pmatrix}
1&0&1&0&0&0\\
0&1&1&0&0&1\\
0&0&0&1&0&1\\
0&0&0&0&1&1
\end{pmatrix}.
\end{align*}
This code exhibits:
\begin{itemize}
\item [(1)] Local protection codes:
\begin{itemize}
\item $\mathcal{C}|_{\{1,2,3\}}=\{\mathbf{c}_1,\mathbf{c}_2,\mathbf{c}_1+\mathbf{c}_2\}$;

\item $\mathcal{C}|_{\{2,4,5,6\}}=\{\mathbf{c}_2,\mathbf{c}_4,\mathbf{c}_5,\mathbf{c}_2+\mathbf{c}_4+\mathbf{c}_5\}$.

    \end{itemize}

\item [(2)] It is optimal with parameters $[6,4,2]_q$, verified through Lemma \ref{pun}.
    \end{itemize}
The decomposition aligns with Case (I) in Theorem \ref{Stru}:
$$C=C_1\cup \{(4,\mathbf{c}_4)\}\cup U,$$
where $C=\{(i,\mathbf{c}_i):\;i\in[6]\}$, $C_1=\{(1,\mathbf{c}_1),(2,\mathbf{c}_2),(3,\mathbf{c}_3)\}$ and $U=\{(5,\mathbf{c}_5),(6,\mathbf{c}_6)\}$.

Consider a $(2,2)$-LRC with the following generator matrix:
\begin{align*}
G=(\mathbf{c}_1,\ldots,\mathbf{c}_5)=
\begin{pmatrix}
1&0&1&0&0\\
0&1&1&0&1\\
0&0&0&1&1
\end{pmatrix}.
\end{align*}
This code exhibits:
\begin{itemize}
\item [(1)] Local protection codes:
\begin{itemize}
\item
$\mathcal{C}|_{\{1,2,3\}}=\{\mathbf{c}_1,\mathbf{c}_2,\mathbf{c}_1+\mathbf{c}_2\}$;

\item $\mathcal{C}|_{\{2,4,5\}}=\{\mathbf{c}_2,\mathbf{c}_4,\mathbf{c}_2+\mathbf{c}_4\}$.

    \end{itemize}

\item [(2)] It is optimal with parameters $[5,3,2]_q$, confirmed through Lemma \ref{pun}.
    \end{itemize}
The decomposition follows Case (II) in Theorem \ref{Stru}:
$$C=\mathcal{R}\cup U,$$
where $C=\{(i,\mathbf{c}_i):\;i\in[5]\}$, and $\mathcal{R}$ and $U$ are respectively given by
$$\mathcal{R}=\{(1,\mathbf{c}_1),(2,\mathbf{c}_2),(3,\mathbf{c}_3)\},\;U=\{(4,\mathbf{c}_4),(5,\mathbf{c}_5)\}.$$
\end{example}

\subsection{Rigidity of local protection codes}\label{sub3.2}
Building upon the framework established in Theorem \ref{Stru}, we now derive the following rigidity property for local protection codes, thereby strengthening the result of \cite[Theorem 5]{Prakash2012}.

\begin{theorem}\label{stru-coro1}
Every local protection code of an optimal $(r,\delta)$-LRC is an MDS code with parameters of the form $[n\leq (r+\delta-1),n+1-\delta,\delta]_{q}$.
\end{theorem}

\begin{proof}
Let $\mathcal{C}$ be an optimal $(r,\delta)$-LRC with a generator matrix $G=(\mathbf{c}_1,\ldots,\mathbf{c}_n)$, and let $C$ be the canonical set defined as
$$C=\{(1,\mathbf{c}_1),\ldots,(n,\mathbf{c}_n)\}.$$
Consider any local protection code $\mathcal{C}|_{S_1}$ associated with a subset $S_1\subseteq C$. By Theorem \ref{Stru}, $C$ decomposes into either Case (I) or Case (II). If $C$ decomposes into Case (I), we can require that $C_1$ in Case (I) is exactly $S_1$ due to the proof of Theorem \ref{Stru}. In this case, $S_1$ is an MDS code with parameters of the form $[n_1\leq (r+\delta-1),k_1,\delta]_{q}$ by Theorem \ref{Stru}. Similarly, if $C$ decomposes into Case (II), we can require that $C_1$ in Case (II) is exactly $S_1$. In this case, $S_1$ is also an MDS code with parameters of the form $[n_1\leq (r+\delta-1),k_1,\delta]_{q}$ by Theorem \ref{Stru}. This completes the proof. $\hfill\square$
\end{proof}

\vspace{8pt}

In fact, the decomposition in Case II of Theorem \ref{Stru} occurs exclusively when $r|(k-1)$. To the best of our knowledge, this decomposition has not been previously established. We now provide a detailed characterization of this scenario in the following proposition.

\begin{proposition}\label{Stru-prop1}
Let $\mathcal{C}$ be an optimal $(r,\delta)$-LRC with parameters $[n,k,d]_{q}$. If $\mathcal{C}$ admits a decomposition in Case (II) of Theorem \ref{Stru}, then $r$ divides $(k-1)$. Furthermore, each local protection code $\mathcal{C}|_{C_i}$ in this decomposition is an MDS code with parameters $[r+\delta-1,r,\delta]_{q}$.
\end{proposition}

\begin{proof}
Case 1: If $r=1$, then the conclusion holds trivially.

Case 2: Consider $r\geq 2$. Assume $C=C_1\cup \ldots\cup C_{t} \cup U$, where $t=\lceil \frac{k}{r}\rceil -1$ and $U\subseteq C$. Adopting the notation from Theorem \ref{Stru}, we define
$$s_i':=|\cup_{j=1}^{i}C_{j}|-|\cup_{j=1}^{i-1}C_{j}|,\;r_i':=\rank(\cup_{j=1}^{i}C_{j})-\rank(\cup_{j=1}^{i-1}C_{j}),$$
where $i \in [t]$ and $C_0=\emptyset$. By Eqs. \eqref{iineq4}, \eqref{iineq9} and \eqref{iineq12} in the proof of Theorem \ref{Stru}, we derive that
\begin{numcases}{}
	 \sum_{i=1}^{j_0} s_i \leq t(r+\delta -1), \label{c1} \\
	\sum_{i=1}^{j_0} s_i =k-1+t(\delta -1).  \label{c2}
\end{numcases}
Combining Eqs. \eqref{c1} and \eqref{c2}, we have
\begin{align}
	k-1+t(\delta -1)\leq t(r+\delta -1) \Rightarrow k-1\leq rt. \label{c3}
\end{align}
Let $k=rm+s$, where $m\in \mathbb{N}^+$ and $0\leq s\leq r-1$. Let us demonstrate that $s=1$.
\begin{itemize}
\item  If $s=0$, then $t=m-1$. Substituting it into Eq. \eqref{c3}, we have
$$rm-1\leq r(m-1),$$
which is impossible for $r\geq 2$. Hence, $s\geq 1$.
\item  For $s\geq 1$, we have $t=m+1$, and Eq. \eqref{c3} becomes
 $$rm+s-1\leq rm \Rightarrow s=1.$$
\end{itemize}
Thus, $k=rm+1$, i.e., $r|(k-1)$.

From Eq. \eqref{c2}, we obtain
$$\sum_{i=1}^{j_0} s_i' =k-1+t(\delta -1)=t(r+\delta -1),$$
where the last equality follows from $r|(k-1)$. By Eq. \eqref{iineq1}, this equality holds only if $s_i'=r+\delta -1 $ for all $i \in [t]$. Consequently, each $\mathcal{C}|_{C_i}$ is an MDS code with parameters $[r+\delta-1,r,\delta]_{q}$. $\hfill\square$
\end{proof}

\vspace{8pt}

Using a similar approach, we present an alternative method to recover the following result, originally established in \cite[Theorem~5]{Prakash2012}.

\begin{corollary}\label{stru-corox}
Let $\mathcal{C}$ be an optimal $(r,\delta)$-LRC with parameters $[n,k,d]_{q}$, where $r|k$. Then, every local protection code of $\mathcal{C}$ is an MDS code with parameters $[r+\delta-1,r,\delta]_{q}$.
\end{corollary}

\begin{proof}
Let $G=(\mathbf{c}_1,\ldots,\mathbf{c}_n)$ be a generator matrix of $\mathcal{C}$. Define the following indexed set $$C=\{(1,\mathbf{c}_1),\ldots,(n,\mathbf{c}_n)\}.$$
Proposition \ref{Stru-prop1} implies that $\mathcal{C}$ only decomposes into Case (I) of Theorem \ref{Stru}. Let $t=\frac{k}{r}$ and consider the following decomposition:
$$C=C_1\cup\ldots\cup C_{t-1}\cup \{(j_1,\mathbf{c}_{j_1}),\ldots,(j_{s_t},\mathbf{c}_{j_{s_t}})\}\cup U.$$
Adopting the notation from Theorem \ref{Stru}, we define
\begin{align*}
s_i:=|\cup_{j=1}^i C_j|-|\cup_{j=1}^{i-1} C_j|,\;r_i:=\rank(\cup_{j=1}^i C_j)-\rank(\cup_{j=1}^{i-1} C_j),\ i \in [t],
\end{align*}
with $C_0:=\emptyset$ and $C_{i_0+1}:=S$. By Eq. \eqref{ineq13}, we derive that
\begin{align}
\sum_{i=1}^{t}s_i=k-1+(t-1)(\delta-1).
\end{align}
Observe that
$$k-1+(t-1)(\delta-1)=(t-1)(r+\delta -1)+r-1.$$
Thus, the equation becomes
\begin{align}
\sum_{i=1}^{t}s_i=(t-1)(r+\delta -1)+r-1.
\end{align}
By Eqs. \eqref{ineq1} and \eqref{ineq3}, this equality holds if and only if
$$s_i = r+\delta -1,\;s_{t}=r-1,$$
for all $i\in [t-1]$. Consequently, each local protection code $\mathcal{C}|_{C_i}$ is an MDS code with parameters $[r+\delta-1,r,\delta]_{q}$. Since $\mathcal{C}$ decomposes exclusively into Case (I) of Theorem \ref{Stru} and the choice of $\mathcal{C}|_{C_1}$ is arbitrary among local protection codes, every local protection code of $\mathcal{C}$ is an MDS code with parameters $[r+\delta-1,r,\delta]_{q}$. $\hfill\square$
\end{proof}

\begin{remark}
A critical observation from the preceding proof is that any two distinct local protection codes $\mathcal{C}|_{S_1}$ and $\mathcal{C}|_{S_2}$ must satisfy either $S_1=S_2$ or $S_1\cap S_2=\emptyset$. This induces a partition of $C$:
$$C=\cup_{i=1}^m C_i,$$
where each $C_i$ satisfies that $\mathcal{C}|_{C_i}$ is a local protection code. Remarkably, this structural property recovers the foundational result in \cite[Theorem 9]{Song2014}, demonstrating consistency with the decomposition theorem presented in Theorem \ref{Stru}.
\end{remark}

Finally, we formalize a structural invariant called the mutual rigidity property of local protection codes, which characterizes their mutual constraints.

\begin{corollary}\emph{(Mutual Rigidity Constraints)}\label{rigidity}
Let $\mathcal{C}$ be an optimal $(r,\delta)$-LRC with parameters $[n,k,d]_{q}$. Suppose $\{C_i\subseteq C:\;i\in [m]\}$ are local protection codes satisfying the following conditions:
\begin{itemize}
\item [(1)] $\rank(\cup_{i=1}^m C_i)<k$;

\item [(2)]
$\rank(\cup_{j=1}^i C_j)-\rank(\cup_{j=1}^{i-1} C_j)\geq 1$ for all $i\in [m]\backslash \{1\}$.
\end{itemize}
Then, for all $i\in [m-1]$, the following rank-cardinality relation holds:
$$\rank(\cup_{j=1}^{i+1}C_i)-\rank(\cup_{j=1}^{i}C_i)=|\cup_{j=1}^{i+1}C_i|-|\cup_{j=1}^{i}C_i|-\delta-1.$$
\end{corollary}

\begin{proof}
Combining Theorem \ref{Stru} and Lemma \ref{Stru-lemm1}, we obtain the desired result. $\hfill\square$
\end{proof}

\section{Construction method and general criterion for optimal quantum $(r,\delta)$-LRCs}\label{quantum}

\subsection{Optimal Quantum $(r,\delta)$-LRCs Construction}
Applying the decomposition theorem in Theorem \ref{Stru}, we can eliminate the condition $n-k\geq \lceil \frac{k}{r}\rceil(\delta-1)$ required in Lemma \ref{quanum-cl}. This leads to the following theorem.

\begin{theorem}\emph{(Optimal Quantum $(r,\delta)$-LRCs Construction)}\label{stru-corobound}
For an optimal $(r,\delta)$-LRC with parameters $[n,k,d]_{q^{2}}$ (resp. $[n,k,d]_{q}$), the following statements hold:
\begin{itemize}
\item [(1)] $n-k\geq \lceil \frac{k}{r}\rceil(\delta-1)$, i.e., $d\geq \delta$;

\item [(2)] If it is Hermitian dual-containing (resp. Euclidean dual-containing), then the induced quantum code is an optimal quantum $(r,\delta)$-LRC.
\end{itemize}
\end{theorem}

\begin{proof}
Let $\mathcal{C}$ be an optimal $(r,\delta)$-LRC with parameters $[n,k,d]_{q^{2}}$ (resp. $[n,k,d]_{q}$). By Lemma \ref{quanum-cl}, it suffices to prove $n-k\geq \lceil \frac{k}{r}\rceil(\delta-1)$.

If $r=1$, then $n=k(\delta-1)$, directly yielding $n-k=\lceil \frac{k}{r}\rceil(\delta-1)$.

If $\lceil \frac{k}{r}\rceil=1$, then $\mathcal{C}$ is an MDS code, and the inequality trivially holds.

If $\lceil \frac{k}{r}\rceil>1$, then we will apply Theorem \ref{Stru} as follows. If $\mathcal{C}$ admits a decomposition of Case (I) in Theorem \ref{Stru}, i.e.,
$C=C_1\cup\ldots\cup C_{t-1}\cup \{(j_1,\mathbf{c}_{j_1}),\ldots,(j_{s_t},\mathbf{c}_{j_{s_t}})\}\cup U$ such that there exists a subset $T\subseteq C$ satisfying $\{(j_1,\mathbf{c}_{j_1}),\ldots,(j_{s_t},\mathbf{c}_{j_{s_t}})\}\subseteq T$, $\mathcal{C}|_{T}$ is a local protection code, and the conditions (1)-(4) of Case (I) in
Theorem \ref{Stru} hold. Denote $C_t=\mathcal{C}|_{T}$. Since $\mathcal{C}|_{C_i}$ is an MDS code with parameters $[n_i\leq (r+\delta-1),k_i,\delta]$, we can assume that its parity-check matrix to be a $(\delta-1)\times n_i$ matrix $H_i$. Define
\begin{align}
s_i:=|\cup_{j=1}^i C_j|-|\cup_{j=1}^{i-1} C_j|, \label{count3}
\end{align}
where $i\in [t]$ and $C_0:=\emptyset$. Reorder the columns of the generator matrix $G$ of the code $\mathcal{C}$ such that for any $ (j,c_j)\in C_i$, the index $j$ satisfies
$$j\leq \sum_{l=1}^i s_l.$$
By Case (I) of Theorem \ref{Stru}, we have $s_i\geq \delta$. Let $I_i=\{i\in [n]:\;(i,\mathbf{c}_i)\in C_i \}$. Let us construct the extended $(\delta-1)\times n$ matrix $H_i'$ as follows:
\begin{itemize}
\item For $j\in I_j$, the $j$th column of $H_i'$ equals to the $j$th column of $H_i$;

\item For $j\notin I_j$, the $j$th column of $H_i'$ equals to $\mathbf{0}$.
\end{itemize}

Define the block matrix as follows:
\begin{align}
H_i''=\begin{pmatrix}
H_1'\\
\vdots\\
H_i'
\end{pmatrix},
\end{align}
where $H_0'':=O_{(\delta-1)\times n}$ is the $(\delta-1)\times n$ zero matrix. By the definition of $H_i'$, the row vectors of $H_i'$ are contained in $\mathcal{C}^{\perp_\text{E}}$. Since $s_i\geq \delta$, the support condition ensures
\begin{align}
|\text{Supp}(H_{i+1}'')\backslash \text{Supp}(H_{i}'')|\geq \delta, \;\ 0\leq i\leq t. \label{supp}
\end{align}
The MDS property of $\mathcal{C}|_{C_i}$ guarantees that any collection of $\delta-1$ columns of $H_i$ is linearly independent. Combining this with Eq. \eqref{supp}, all rows of $H_t''$ are linearly independent. Since the space spanned by the row vectors of $H_t''$ is contained in $\mathcal{C}^{\perp_\text{E}}$, and $H_t''$ has $t(\delta-1)$ row vectors, we derive that $$\dim(\mathcal{C}^{\perp_\text{E}})\geq t(\delta-1)\Rightarrow n-k\geq t(\delta-1).$$

If $\mathcal{C}$ admits the decomposition of Case (II) in Theorem \ref{Stru}, a parallel argument applies, yielding the same inequality $n-k\geq t(\delta-1)$.

Therefore, we complete the proof. $\hfill\square$
\end{proof}

\vspace{8pt}

In fact, there exist optimal $(r,\delta)$-LRCs with parameters $[n,k,d]_{q}$ that achieves equality in the bound $n-k \geq \lceil \frac{k}{r} \rceil (\delta-1)$.
For instance, consider the linear code $\mathcal{C}$ with parameters $[4,2,2]_q$, whose generator matrix is given by
\begin{align*}
G=(\mathbf{c}_1,\mathbf{c}_1,\mathbf{c}_2,\mathbf{c}_2)=
\begin{pmatrix}
1&1&0&0\\
0&0&1&1
\end{pmatrix}.
\end{align*}
This code is an optimal $(1,2)$-LRC, achieving $n-k=\lceil \frac{k}{r}\rceil(\delta-1)=2$, thereby saturating the bound $n-k=\lceil \frac{k}{r}\rceil(\delta-1)$.

\subsection{General criterion for determining optimal quantum $(r,\delta)$-LRCs}

Based on the structural insights from the decomposition theorem in Theorem \ref{Stru}, we will present a complete characterization of a class of optimal $(r,\delta)$-LRCs in the following theorem. To this end, we first introduce the concept of minimal decomposition.

\begin{definition}\emph{(Minimal Decomposition)}\label{mini-dec}
Let $\mathcal{C}$ be an optimal $(r, \delta)$-LRC, and let $C$ be the set defined in Eq. \eqref{setC}. If $C$ admits a decomposition $C=C_1\cup \ldots\cup C_{t}$, where $t=\lceil \frac{k}{r}\rceil$, satisfying that $\mathcal{C}|_{C_i}$ is a local protection code for all $i \in [t]$, and $C_i\cap C_j=\emptyset$ for all $|i-j|> 1$, then we say that $\mathcal{C}$ admits a \emph{minimal decomposition}.
\end{definition}

\begin{theorem}\label{Stru-apply}
Let $\mathcal{C}$ be an optimal $(r,\delta)$-LRC. If $\mathcal{C}$ admits a minimal decomposition, then $\mathcal{C}$ possesses a parity-check matrix of the following form:
\begin{align}
H=\begin{pmatrix}
A_{1,1}& A_{1,2} & & && & &\\
& A_{2,1}&A_{2,2} & && & &\\
& &A_{3,1} &A_{3,2} && & &\\
& & & &\ddots & & &\\
&&& && A_{t-1,1}&A_{t-1,2}& \\
&&& && & A_{t,1}&A_{t,2} \\
B_{1}& B_{2}&\cdots &\cdots&\cdots&\cdots & B_{t}&B_{t+1}
\end{pmatrix}, \label{pari}
\end{align}
where
\begin{itemize}
\item[(i)] $A_{i,j}$ is an $(\delta-1)\times u_i$ matrix for some $u_i\in \mathbb{N}^+$;

\item[(ii)] $B_i$ is an $l \times v_i$ matrix for some $l,v_i\in \mathbb{N}^+$;

\item [(iii)] Each code $\mathcal{C}_{i}$ with a generator matrix $(A_{i,1},A_{i,2})$ is an MDS code with parameters $[n_i\leq (r+\delta-1),\delta-1,n_i-\delta+2]_{q}$;

\item[(iv)]  $d(\mathcal{C})=\delta+l$.
\end{itemize}
Conversely, any code $\mathcal{C}$ whose parity-check matrix satisfies the form given in Eq. (\ref{pari}) and meets conditions (i)-(iv) above must be an optimal $(r,\delta)$-LRC which admits a minimal decomposition.
\end{theorem}

\begin{proof}
Let $\mathcal{C}$ be an optimal $(r,\delta)$-LRC with a minimal decomposition $C=C_1\cup \ldots\cup C_{t}$, where $t=\lceil \frac{k}{r}\rceil$, satisfying that $\mathcal{C}|_{C_i}$ is a local protection code for $i \in [t]$, and $C_i\cap C_j=\emptyset$ for all $|i-j|> 1$. By Corollary \ref{stru-coro1}, each $\mathcal{C}|_{C_i}$ is an MDS code with parameters $[n_i\leq (r+\delta-1),k_i,\delta]_{q}$. Let $H_i$ be an $(\delta-1)\times n_i$ parity-check matrix of $\mathcal{C}|_{C_i}$. The disjointness condition $C_i\cap C_j=\emptyset$ for all $|i-j|> 1$ allows us to partition each $H_i$ into block matrices $(A_{i,1},A_{i,2})$ such that the row vectors of the following matrix belong to $\mathcal{C}^{\perp_\text{E}}$:
\begin{align}
\begin{pmatrix}
\widehat{H_1}\\
\vdots\\
\widehat{H_t}
\end{pmatrix}=\begin{pmatrix}
A_{1,1}& A_{1,2} & & && & &\\
& A_{2,1}&A_{2,2} & && & &\\
& &A_{3,1} &A_{3,2} && & &\\
& & & &\ddots & & &\\
&&& && A_{t-1,1}&A_{t-1,2}& \\
&&& && & A_{t,1}&A_{t,2}
\end{pmatrix}, \label{arr}
\end{align}
where $A_{i,j}$ are $(\delta-1)\times u_i$ matrices.
Note that the number of repeated indices between $\widehat{H_i}$ and $\widehat{H_{i+1}}$ is less than $n_i - (\delta - 1)$. Consequently, the matrix given by  Eq. \eqref{arr} has rank $t(\delta - 1)$.
Extending it to the form given in Eq. \eqref{pari} by appending matrices $B_1,\ldots,B_t$ yields that
$$\rank(H)=t(\delta-1)+l.$$
Hence, the dimension of $\mathcal{C}$ is $k=n-t(\delta-1)-l$. Since $\mathcal{C}$ is an optimal $(r,\delta)$-LRC, the minimum distance satisfies that
\begin{align*}
d(\mathcal{C})&=n-k+1-(t-1)(\delta-1)\\
&=t(\delta-1)+l+1-(t-1)(\delta-1)\\
&=[t(\delta-1)-(t-1)(\delta-1)]+(l+1)\\
&=\delta+l.
\end{align*}

Conversely, suppose $\mathcal{C}$ has a parity-check matrix of the form Eq. \eqref{pari} with $d(\mathcal{C})=\delta+l$. By Lemma \ref{local},  $\mathcal{C}$ is an $(r,\delta)$-LRC. Substituting $d(\mathcal{C})=\delta+l$ into the equality
$$\delta+l=n-k+1-(\lceil \frac{k}{r}\rceil -1)(\delta-1),$$
we verify that $\mathcal{C}$ is an optimal $(r,\delta)$-LRC. Furthermore, the specific form of the parity-check matrix of $\mathcal{C}$ ensures the existence of a minimal decomposition. This completes the proof. $\hfill\square$
\end{proof}

\vspace{8pt}

According to Theorem \ref{Stru-apply}, we obtain the following corollary.

\begin{corollary}\label{stru-coro2}
Let $\mathcal{C}$ be an optimal $(r,\delta)$-LRC, and let $C$ be the set defined in Eq. (\ref{setC}). Suppose $C=C_1\cup \ldots\cup C_{t}$, where $t=\lceil \frac{k}{r}\rceil$, is a minimal decomposition. If $C_i\cap C_j=\emptyset$ for all $1\leq i\neq j\leq t$, then $\mathcal{C}$ possesses a parity-check matrix of the following form:
\begin{align}
H=\begin{pmatrix}
A_{1} & & &\\
& A_{2} &  &\\
&  & \ddots &\\
&  &  &A_t\\
B_{1}& B_{2} & \cdots &B_t
\end{pmatrix}, \label{pari-sim}
\end{align}
where
\begin{itemize}
\item[(i)] $A_{i,j}$ is an $(\delta-1)\times u_i$ matrix;

\item[(ii)] $B_i$ is an $l \times v_i$ matrix for some $l,u_i,v_i\in \mathbb{N}^+$;

\item [(iii)] Each code $\mathcal{C}_{i}$ with a generator matrix $(A_{i,1},A_{i,2})$ is an MDS code with parameters $[n_i\leq (r+\delta-1),\delta-1,n_i-\delta+2]_{q}$;

\item[(iv)]  $d(\mathcal{C})=\delta+l$.
\end{itemize}
Conversely, any code $\mathcal{C}$ whose parity-check matrix satisfies the form given in Eq. (\ref{pari}) and meets conditions (i)-(iv) above must be an optimal $(r,\delta)$-LRC.
\end{corollary}

\begin{proof}
Let $\mathcal{C}$ be an optimal $(r,\delta)$-LRC satisfying $C=C_1\cup \ldots\cup C_{t}$, where $t=\lceil \frac{k}{r}\rceil$, $\mathcal{C}|_{C_i}$ is a local protection code for $i \in [t]$, and $C_i\cap C_j=\emptyset$ for all $1\leq i\neq j\leq t$. The disjointness condition $C_i\cap C_j=\emptyset$ for all $1\leq i\neq j\leq t$ ensures that the parity-check matrix decouples into the form specified in Eq. \eqref{pari-sim}.
The remainder of the proof then follows identically to the argument in Theorem \ref{Stru-apply}.
$\hfill\square$
\end{proof}

\begin{remark}\label{remark4}
The parity-check matrix with the form of Eq. (\ref{pari-sim}) appears frequently in constructions of optimal LRCs (see, e.g., \cite{Kong2021}).
In this work, we prove that for a class of optimal LRCs, the parity-check matrix must necessarily take this form.
To the best of our knowledge, the interleaved form of the parity-check matrix presented in Theorem \ref{Stru-apply} has not been reported in the prior literature.
\end{remark}

In the following theorem, we provide a full characterization of a class of optimal quantum $(r,\delta)$-LRCs, and establish a general criterion for these codes.

\begin{theorem}\label{Stru-quan}
If $\mathcal{C}$ is an optimal $(r,\delta)$-LRC that admits a minimal decomposition, and it induces an optimal quantum $(r,\delta)$-LRC,
then $\mathcal{C}$ possesses a parity-check matrix $H$ of the form in Eq. \eqref{pari} with the property that all codes generated by either
\begin{align}
G_i=\begin{pmatrix}
A_{i,1} &A_{i,2}\\
B_i & B_{i+1}
\end{pmatrix}\ \mathrm{or} \
G_j'=\begin{pmatrix}
A_{j,2}\\
A_{j+1,1}
\end{pmatrix}\label{hermi}
\end{align}
for $i\in [t]$ and $j\in [t-1]$, are either all Hermitian self-orthogonal or all Euclidean self-orthogonal.

Conversely, if $\mathcal{C}$ has such a parity-check matrix of the form in Eq. \eqref{pari} satisfying that all codes generated by Eq. \eqref{hermi} are all Hermitian self-orthogonal or all Euclidean self-orthogonal, then $\mathcal{C}$ is an optimal $(r,\delta)$-LRC that admits a minimal decomposition, and it induces an optimal quantum $(r,\delta)$-LRC.
\end{theorem}

\begin{proof}
Let $\mathcal{C}$ be an optimal $(r,\delta)$-LRC. Assume $C=C_1\cup \ldots\cup C_{t}$, where $t=\lceil \frac{k}{r}\rceil$, is a minimal decomposition. Suppose $\mathcal{C}$ induces an optimal quantum $(r,\delta)$-LRC. By Definition \ref{quan-LRC}, $\mathcal{C}$ must be Hermitian dual-containing or Euclidean dual-containing. Let $H$ be a parity-check matrix of $\mathcal{C}$ with the form in Eq. \eqref{pari}.
If $\mathcal{C}$ is Hermitian dual-containing, then the rows $\mathbf{n}_i$ and $\mathbf{n}_j$ of its parity-check matrix $H$ satisfy $\langle \mathbf{n}_i,\mathbf{n}_j\rangle_\text{H}=0$ for all $i,j$. This implies that the codes generated by the matrices in Eq. \eqref{hermi} are Hermitian self-orthogonal. Similarly, if $\mathcal{C}$ is Euclidean dual-containing, one can derive that the codes generated by the matrices in Eq. \eqref{hermi} are Euclidean self-orthogonal.

Conversely, suppose the codes with generator matrices given in Eq. \eqref{hermi} are Hermitian self-orthogonal (resp. Euclidean self-orthogonal), then one can check that $\mathcal{C}$ is Hermitian dual-containing (resp. Euclidean dual-containing). By Theorems \ref{stru-corobound} and \ref{Stru-apply}, $\mathcal{C}$ induces an optimal quantum $(r,\delta)$-LRC. $\hfill\square$
\end{proof}

\vspace{6pt}

According to Theorem \ref{Stru-quan}, we obtain the following corollary.

\begin{corollary}\label{quan-coro3}
Let $\mathcal{C}$ be an optimal $(r,\delta)$-LRC, and let $C$ be the set defined in Eq. (\ref{setC}). Suppose $C=C_1\cup \ldots\cup C_{t}$, where $t=\lceil \frac{k}{r}\rceil$, is a minimal decomposition. If $C_i\cap C_j=\emptyset$ for all $1\leq i\neq j\leq t$, then $\mathcal{C}$ induces an optimal quantum $(r,\delta)$-LRC if and only if its parity-check matrix $H$ defined in Eq. \eqref{pari-sim} satisfies the condition that all codes generated by
\begin{align}
G_i=\begin{pmatrix}
A_{i}\\
B_i
\end{pmatrix}
\label{hermi1}
\end{align}
for $i\in [t]$ are either all Hermitian self-orthogonal or all Euclidean self-orthogonal.
\end{corollary}

\begin{proof}
By applying Corollary \ref{stru-coro2} and Theorem \ref{Stru-quan}, the desired result is obtained. $\hfill\square$
\end{proof}

\begin{remark}\label{remark1}
Corollary \ref{stru-coro2}  and Theorem \ref{Stru-quan} collectively establish a systematic methodology for constructing optimal quantum $(r,\delta)$-LRCs. Specifically:
\begin{itemize}
\item Step 1: Designing the local protection codes $\mathcal{C}|_{C_i}$ as MDS codes with parameters $[n_i\leq r+\delta-1,\delta-1,n_i-\delta+2]_{q}$;

\item Step 2: Putting the parity-check matrices of $\mathcal{C}|_{C_i}$ together to obtain a matrix of the form given in Eq. (\ref{arr});

\item Step 3: Appending the matrices $B_1,\ldots, B_t$ to get a matrix of the form given in Eq. (\ref{pari}) that satisfies all conditions of Theorem \ref{Stru-quan}.
\end{itemize}
In Section \ref{construction}, we demonstrate this framework by explicit constructions of optimal quantum $(r,\delta)$-LRCs, validating both the feasibility and scalability of the approach.
\end{remark}

\section{Three infinite families of optimal quantum $(r,\delta)$-LRCs}\label{construction}

Using Theorem \ref{Stru-quan}, we construct in this section three infinite families of optimal quantum $(r,\delta)$-LRCs with distinct parameters $n,k,d, r,\delta$. These families extend beyond the optimal quantum $(r,\delta)$-LRCs given in \cite[Proposition 34]{Galindo2024}. The construction of optimal classical $(r,\delta)$-LRCs has attracted significant attention, and the first infinite family construction comes from \cite{Tamo2014}. Subsequent work has gained significant momentum (see, e.g., \cite{Rawat2013,Papailiopoulos2014,Ernvall2015,Tamo2016,Micheli2019,Kolosov2018,Kamath2014}). In the following content, we focus on constructing optimal quantum $(r,\delta)$-LRCs by leveraging optimal classical $(r,\delta)$-LRCs.

\subsection{The first family of optimal quantum $(r,\delta)$-LRCs}

\begin{theorem}\label{Opt-1}
Let $\mathbb{F}=\mathbb{F}_{q^{2}}$ with $q\geq 5$, and let $u,v,t\in \mathbb{N}^+$ satisfy $v\leq u$ and $u+v\leq \lfloor \frac{q-1}{2}\rfloor -1$. Fix primitive $(q-1)$th roots of unity $\omega_1,\ldots,\omega_t\in \mathbb{F}$. Define a linear code $\mathcal{C}$ with a parity-check matrix as follows:
\begin{align}\label{pari-1}
H=\begin{pmatrix}
A_{1} & & &\\
& A_{2} &  &\\
&  & \ddots &\\
&  &  &A_t\\
B_{1}& B_{2} & \cdots &B_t
\end{pmatrix}, 
\end{align}
where
\begin{align}
A_i=\begin{pmatrix}
1 & \omega_i & \cdots&  \omega_i^{q-2}\\
\vdots & \vdots & \vdots&  \vdots\\
1 & \omega_i^{u} & \cdots&  \omega_i^{u(q-2)}
\end{pmatrix}\ \mathrm{and}\ B_i=\begin{pmatrix}
1 & \omega_i^{u+1} & \cdots&  \omega_i^{(u+1)(q-2)}\\
\vdots & \vdots & \vdots&  \vdots\\
1 & \omega_i^{u+v} & \cdots&  \omega_i^{(u+v)(q-2)}
\end{pmatrix} \label{qpari-1}
\end{align}
for $i\in [t]$ and the blank entries in matrix $H$ are all zeros. Then, the following statements hold:
\begin{itemize}
\item [(1)] $\mathcal{C}$ is an optimal $(r,\delta)$-LRC with parameters $[t(q-1),t(q-1)-tu-v,u+v+1]_{q^2}$, where $(r,\delta)=(q-1-u,u+1)$.

\item [(2)] $\mathcal{C}$ is Hermitian dual-containing, inducing an optimal quantum $(r,\delta)$-LRC with parameters
\begin{align}
\bigg[\mspace{-4mu}\bigg[t(q-1),t(q-1)-2tu-2v, u+v+1\bigg]\mspace{-4mu}\bigg]_{q}. \label{quan-1}
\end{align}
\end{itemize}
\end{theorem}

\begin{proof}
Let $A_i=(\omega_i^{j(l-1)})_{1\leq j\leq u,1\leq l\leq q-1}$ and let $B_i=(\omega_i^{(u+j)(l-1)})_{1\leq j\leq v,1\leq l\leq q-1}$, where $i\in[t]$. Then, the matrix $H$ defined in Eq. \eqref{pari-1} matches the form specified in Eq. \eqref{pari-sim} of Corollary \ref{stru-coro2}. To prove that $\mathcal{C}$ is an optimal $(r,\delta)$-LRC, it suffices to verify that the matrices $A_i$ and $B_i$ satisfy the conditions of Corollary \ref{stru-coro2}. By definition, $A_{i}$ is an $(\delta-1)\times (q-1)$ matrix, $B_i$ is an $v \times (q-1)$ matrix, and the code $\mathcal{C}_{i}$ generated by $A_{i}$ is an MDS code with parameters $[q-1,\delta-1,q-\delta+1]_{q^{2}}$, where $\delta=u+1$. Hence, we only need to verify that $d(\mathcal{C})=u+v+1$.

Obviously, the first $u+v+1$ columns of $H$ are linearly dependent. Thus, $d(\mathcal{C})\leq u+v+1$. Let $E=\{\mathbf{v}_i:\;i\in[t(q-1)]\}$ be the set of column vectors of $H$. For any subset $T\subseteq E$ with $|T|=u+v$, partition $T$ as $T=\cup_{j=1}^t T_j$, where $T_j\subseteq \{\mathbf{v}_i:\;(j-1)(q-1)+1\leq i\leq j(q-1)\}$. If $|T_1|\leq u$, then $T_1\notin \langle T_j:\;1< j\leq t\rangle$ due to the form of $H$. If $|T_1|>u$, then $|T_j|\leq u-1$ for $j>1$, since $|T|=u+v$ and $v\leq u$. Thus, $\cup_{j=2}^t T_j\notin \langle T_1\rangle$. Both cases ensure that all vectors in $T$ are linearly independent. Hence, any $u+v$ columns of $H$ are linearly independent. This implies $d(\mathcal{C})\geq u+v+1$. Therefore, $d(\mathcal{C})=u+v+1$.

Next, we use Corollary \ref{quan-coro3} to prove that the induced quantum code from $\mathcal{C}$ is an optimal quantum $(r,\delta)$-LRC. Let $1\leq i\leq tu+v$.
Denote by $\mathbf{n}_{i}$ the $i$th row vector of $H$. Then, we obtain
\begin{numcases}{} \langle\mathbf{n}_{(i-1)u+j},\mathbf{n}_{(i-1)u+j'}\rangle_{\mathrm{H}}=\sum_{l=1}^{q-1}
\omega_i^{(j+qj')(l-1)},\;1\leq i \leq t,1\leq j,j'\leq q-1, \label{inner1} \\	\langle\mathbf{n}_{(i-1)u+j},\mathbf{n}_{tu+j'}\rangle_{\mathrm{H}}=\sum_{l=1}^{q-1}
\omega_i^{(j+q(u+j'))(l-1)},\;1\leq i \leq t,1\leq j,j'\leq q-1, \label{inner2} \\
\langle\mathbf{n}_{tu+j},\mathbf{n}_{tu+j'}\rangle_{\mathrm{H}}=\sum_{i=1}^{t}\sum_{l=1}^{q-1}
\omega_i^{(u+j+q(u+j'))(l-1)},\;1\leq i \leq t,1\leq j,j'\leq v. \label{inner3}
\end{numcases}
Since $u+v\leq \lfloor \frac{q-1}{2}\rfloor -1$, it follows that $2(u+v)\leq q-3$. Thus, $\omega_i^{(u+j+q(u+j'))}\neq 1$ for $1\leq i \leq t$ and $1\leq j,j' \leq v$. Hence, $\langle\mathbf{n}_{tu+j},\mathbf{n}_{tu+j'}\rangle_{\mathrm{H}}=0$ by Eq. \eqref{inner3}. Similarly, we have $\langle\mathbf{n}_{(i-1)u+j},\mathbf{n}_{(i-1)u+j'}\rangle_{\mathrm{H}}=\langle\mathbf{n}_{(i-1)u+j},\mathbf{n}_{tu+j'}\rangle_{\mathrm{H}}=0$ for $1\leq i \leq t$ and $1\leq j,j'\leq q-1$ by Eqs. \eqref{inner1} and \eqref{inner2}. Since the set $\{\mathbf{n}_{(i-1)u+j}:\;i \in [t],1\leq j \leq q-1\}\cup \{\mathbf{n}_{tu+j}:\;i \in [t],1\leq j\leq v\}$ comprises all row vectors of $H$, we obtain $\langle\mathbf{n}_{i},\mathbf{n}_{j}\rangle_{\mathrm{H}}=0$ for any row vectors $\mathbf{n}_{i}$ and $\mathbf{n}_{j}$. Consequently, the codes generated by the matrices
\begin{align}
G_i=\begin{pmatrix}
A_{i}\\
B_i
\end{pmatrix},
\label{hermi1} i\in [t]
\end{align}
are Hermitian self-orthogonal. By Corollary \ref{quan-coro3}, $\mathcal{C}$ induces an optimal quantum $(r,\delta)$-LRC with parameters given by Eq. \eqref{quan-1}.
$\hfill\square$
\end{proof}

\begin{remark}\label{remark2}
The first family of optimal quantum $(r,\delta)$-LRCs constructed in Theorem \ref{Opt-1} has super-linear length, while the subsequent two families presented in Theorems \ref{Opt-2} and \ref{Opt-3} do not. This distinctive length property guarantees that there is no overlap in parameters between the first family and the latter two.
\end{remark}

Let us give an example to illustrate Theorem \ref{Opt-1}.

\begin{example}\label{quan-example1}
Let $\mathbb{F}=\mathbb{F}_{11^2}$. Set $u=v=t=2$, and let $\omega_1,\omega_2\in \mathbb{F}$ be primitive $10$th roots of unity. Let $\mathcal{C}$ be a linear code with the following parity-check matrix:
\begin{align*}
H=\begin{pmatrix}
1 & \omega_1 & \cdots&  \omega_1^{9}& 0& 0 & \cdots &0\\
1 & \omega_1^2 & \cdots&  \omega_1^{2\times 9}&0 & 0& \cdots &0\\
0& 0 & \cdots& 0& 1 & \omega_2 & \cdots&  \omega_2^{9}\\
0& 0 & \cdots& 0& 1 & \omega_2^2 & \cdots&  \omega_2^{2\times 9}\\
 1 & \omega_1^3 & \cdots&  \omega_1^{3\times 9}&  1 & \omega_2^3 & \cdots&  \omega_2^{3\times 9}\\
 1 & \omega_1^4 & \cdots&  \omega_1^{4\times 9}&  1 & \omega_2^4 & \cdots&  \omega_2^{4\times 9}
\end{pmatrix}.
\end{align*}
By Theorem \ref{Opt-1}, $\mathcal{C}$ is an optimal $(8,3)$-LRC with parameters $[20,14,5]_{11^2}$, inducing an optimal quantum $(8,3)$-LRC with parameters $[\mspace{-2mu}[20,8,5]\mspace{-2mu}]_{11}$.

Furthermore, we give a decomposition of Case (I) in Theorem \ref{Stru}. Let $G=(\mathbf{c}_1,\ldots,\mathbf{c}_{20})$ be a generator matrix of $\mathcal{C}$ satisfying $GH^\top=O_{14\times 6}$. Then, $G$ can be given by:
\begin{align}
\left (
\begin{array}{cccc cccc cccc cccc cccc}
1 & 0 & 0 & 0 & 0 & 0 & 0 & 0 & 8 & 5 & 0 & 0 & 0 & 0 & 0 & 0 & 1 & 8 & 8 & 9 \\
0 & 1 & 0 & 0 & 0 & 0 & 0 & 0 & 4 & 5 & 0 & 0 & 0 & 0 & 0 & 0 & 8 & 10 & 0 & 2 \\
0 & 0 & 1 & 0 & 0 & 0 & 0 & 0 & 4 & 1 & 0 & 0 & 0 & 0 & 0 & 0 & 4 & 7 & 4 & 4 \\
0 & 0 & 0 & 1 & 0 & 0 & 0 & 0 & 3 & 10 & 0 & 0 & 0 & 0 & 0 & 0 & 6 & 8 & 1 & 5 \\
0 & 0 & 0 & 0 & 1 & 0 & 0 & 0 & 8 & 8 & 0 & 0 & 0 & 0 & 0 & 0 & 7 & 7 & 2 & 6 \\
0 & 0 & 0 & 0 & 0 & 1 & 0 & 0 & 2 & 1 & 0 & 0 & 0 & 0 & 0 & 0 & 8 & 5 & 1 & 0 \\
0 & 0 & 0 & 0 & 0 & 0 & 1 & 0 & 3 & 8 & 0 & 0 & 0 & 0 & 0 & 0 & 10 & 0 & 8 & 3 \\
0 & 0 & 0 & 0 & 0 & 0 & 0 & 1 & 2 & 7 & 0 & 0 & 0 & 0 & 0 & 0 &0 &10 & 9 & 4 \\
0 & 0 & 0 & 0 & 0 & 0 & 0 & 0 & 0 & 0 & 1 & 0 & 0 & 0 & 0 & 0 & 1 & 8 & 5 & 3 \\
0 & 0 & 0 & 0 & 0 & 0 & 0 & 0 & 0 & 0 & 0 & 1 & 0 & 0 & 0 & 0 & 8 & 10 & 4 & 7 \\
0 & 0 & 0 & 0 & 0 & 0 & 0 & 0 & 0 & 0 & 0 & 0 & 1 & 0 & 0 & 0 & 4 & 7 & 8 & 5 \\
0 & 0 & 0 & 0 & 0 & 0 & 0 & 0 & 0 & 0 & 0 & 0 & 0 & 1 & 0 & 0 & 6 & 8 & 4 &4\\
0 & 0 & 0 & 0 & 0 & 0 & 0 & 0 & 0 & 0 & 0 & 0 & 0 & 0 & 1 & 0 & 7 & 7 & 10 &3\\
0 & 0 & 0 & 0 & 0 & 0 & 0 & 0 & 0 & 0 & 0 & 0 & 0 & 0 & 0 & 1 & 8 & 5 & 3 &1
\end{array}
\right ).
\end{align}
Partition the set $C=\{(i,\mathbf{c}_i):\;1\leq i\leq 20\}$ as follows. Let $C_1=\{(i,\mathbf{c}_i):\;1\leq i\leq 10\}$ and $U=\{(i,\mathbf{c}_{i}):\;16 \leq i \leq 20\}$. Then,
$$C=C_1\cup \{(11,\mathbf{c}_{11}),\ldots,(15,\mathbf{c}_{15})\}\cup U$$
constitutes a decomposition of Case (I) in Theorem \ref{Stru}.
\end{example}

\subsection{The second family of optimal quantum $(r,\delta)$-LRCs}

In the following theorem, we present the second family of optimal quantum $(r,\delta)$-LRCs.

\begin{theorem}\label{Opt-2}
Let $\mathbb{F}=\mathbb{F}_{q^{2}}$ with $q\geq 7$. Let $s,v,t\in \mathbb{N}^+$ satisfy
\begin{itemize}
\item  $v|(q-1),\;v\geq 6$ and $s\leq \lfloor \frac{v}{2}\rfloor-2$;

\item  $vt\leq q+v-s-3$.
\end{itemize}
Let $\mathcal{C}$ be a linear code with a parity-check matrix:
\begin{align}\label{qpari-2}
H=\begin{pmatrix}
A_{1} & & &\\
& A_{2} &  &\\
&  & \ddots &\\
&  &  &A_t\\
B_{1}& B_{2} & \cdots &B_t
\end{pmatrix},  
\end{align}
where
\begin{align}
A_i=
\left (
\begin{array}{cccc cccc cccc}
x_{i,0} & x_{i,1} & \cdots& x_{i,q-2}&y_0 &y_1&\cdots& y_{v-1} &  &  & & \\
\vdots & \vdots & \cdots& \vdots&\vdots &\vdots&\cdots& \vdots &  &  &  & \\
x_{i,0}^s & x_{i,1}^s & \cdots& x_{i,q-2}^s&y_0^s &y_1^s&\cdots& y_{v-1}^s&  &  & &\\
& & & & y_0 &y_1&\cdots& y_{v-1}&x_{i,0} & x_{i,1} & \cdots& x_{i,q-2} \\
& & & &\vdots &\vdots&\cdots& \vdots&\vdots & \vdots & \cdots& \vdots \\
& & & &y_0^s &y_1^s&\cdots& y_{v-1}^s & x_{i,0}^s & x_{i,1}^s & \cdots & x_{i,q-2}^s
\end{array}
\right )
\end{align}
and
\begin{align*}
B_i=(x_{i,0}^{s+1},\ldots,x_{i,q-2}^{s+1},y_0^{s+1},\ldots,y_{v-1}^{s+1},x_{i,0}^{s+1},\ldots,x_{i,q-2}^{s+1})
\end{align*}
for $i\in [t]$, and the blank entries in matrix $H$ are all zeros.
Here, $x_{i,j}=\lambda_i \omega_i^j$ for $i\in [t]$ and $0\leq j\leq q-2$, and $y_l=\zeta^l$ for $0\leq l\leq v-1$, where
$\omega_1,\ldots, \omega_t\in \mathbb{F}$ are primitive $(q-1)$th roots of unity, $\zeta \in \mathbb{F}$ is a primitive $v$th root of unity, and
$\lambda_i\in \mathbb{F}^\times \backslash \{\omega_i^j:\;1\leq j\leq q-1\}$ for $i\in [t]$.
Then, the following statements hold:
\begin{itemize}
\item [(1)] $\mathcal{C}$ is an optimal $(q+v-s-1,s+1)$-LRC with parameters $[t(2q+v-2),t(2q+v-2)-2st-1, s+2]_{q^2}$.

\item [(2)] $\mathcal{C}$ is Hermitian dual-containing, inducing an optimal quantum $(q+v-s-1,s+1)$-LRC with parameters
\begin{align}
\bigg[\mspace{-4mu}\bigg[t(2q+v-2),t(2q+v-2)-4st-2, s+2\bigg]\mspace{-4mu}\bigg]_{q}. \label{quan-2}
\end{align}
\end{itemize}
\end{theorem}

\begin{proof}
By Lemma \ref{local}, $\mathcal{C}$ is an $(q+v-s-1,s+1)$-LRC.
For $i \in [t]$, set $(B_{2i-1},B_{2i})=B_i$ in Theorem \ref{Stru-apply}, where $B_i$ refers to the matrix defined in this theorem. To apply Theorem \ref{Stru-apply}, we will check the conditions of Theorem \ref{Stru-apply}. The only non-trivial condition is to show $\lceil\frac{k}{r}\rceil=2t$ and $d(\mathcal{C})=s+2$. Since $k=t(2q+v-2)-2st-1=2t(q+v-s-1)-(1+vt)$ and $1+vt\leq q+v-s-2$ by assumption, we have $\lceil\frac{k}{r}\rceil=2t$. Clearly, the first $s+2$ columns of $H$ are linearly dependent. Thus, $d(\mathcal{C})\leq s+2$. Since $\lambda_i\in \mathbb{F}^\times\backslash \{\omega_i^j:\;1\leq j\leq q-1\}$ ensures $x_{i,j}\neq y_l$ for all $1\leq i \leq t,\; 0\leq j\leq q-1,\;0\leq l\leq v-1$, and by using the same technique as in Theorem \ref{Opt-1}, any $s+1$ columns of $H$ are linearly independent. Therefore, $d(\mathcal{C})=s+2$.

Next, we prove that the induced quantum code from $\mathcal{C}$ is an optimal quantum $(r,\delta)$-LRC. Let $1\leq i\leq 2st+1$. Denote by $\mathbf{n}_{i}$ the $i$th row vector of $H$. Since $s\leq \lfloor \frac{v}{2}\rfloor-2$, and by computations analogous to those in Theorem \ref{Opt-1}, all Hermitian inner products satisfy $\langle\mathbf{n}_{i},\mathbf{n}_{j}\rangle_{\mathrm{H}}=0$. Consequently, the codes whose generator matrices are given in Eq. \eqref{hermi} are Hermitian self-orthogonal. By Theorem \ref{Stru-quan}, $\mathcal{C}$ induces an optimal quantum $(r,\delta)$-LRC with parameters specified in Eq. \eqref{quan-2}. $\hfill\square$
\end{proof}

\vspace{6pt}

Let us give an example to illustrate Theorem \ref{Opt-2}.

\begin{example}\label{quan-example2}
Let $\mathbb{F}=\mathbb{F}_{9^2}$ with a primitive element $g$. Set $v=8,\;t=1,\;s=2$, and let $\omega, \zeta \in \mathbb{F}$ be primitive $8$th roots of unity. Choose $\lambda=g$ in Theorem \ref{Opt-2}. Define
\vspace{-4pt}
\begin{align*}
x_{1,i}=x_i=\lambda \omega^{i}\ \mathrm{and}\ y_i=\zeta^{i}
\end{align*}
for $0\leq i\leq 7$. Then, the code $\mathcal{C}$ in Theorem \ref{Opt-2} has a parity-check matrix as follows:
\begin{align}
H=
\left (
\begin{array}{cccc cccc cccc}
x_0 & x_1 & \cdots& x_{7}&y_0 &y_1&\cdots& y_{7}&0 & 0 & \cdots&0 \\
x_0^2 & x_1^2 & \cdots& x_{7}^2&y_0^2 &y_1^2&\cdots& y_{7}^2&0 & 0 &\cdots & 0\\
0& 0& \cdots&0 &y_0 &y_1&\cdots& y_{7}&x_0 & x_1 & \cdots& x_{7}\\
0&0 & \cdots& 0&y_0^2 &y_1^2&\cdots& y_{7}^2&x_0^2 & x_1^2 & \cdots& x_{7}^2\\
x_0^3 & x_1^3 & \cdots& x_{7}^3&y_0^3 &y_1^3&\cdots& y_{7}^3&x_0^3 & x_1^3 & \cdots& x_{7}^3
\end{array}
\right ).
\end{align}
By Theorem \ref{Opt-2}, $\mathcal{C}$ is an optimal $(14,3)$-LRC with parameters $[24,19,4]_{9^2}$, inducing an optimal quantum $(14,3)$-LRC with parameters $[\mspace{-2mu}[24,14,4]\mspace{-2mu}]_{9}$.

Furthermore, let us give a decomposition of Case (I) in Theorem \ref{Stru}. The generator matrix $G=(\mathbf{c}_1,\ldots,\mathbf{c}_{24})$ of $\mathcal{C}$ satisfying $GH^\top=O_{19\times 5}$ has the following block structure:
\begin{align*}
\begin{pmatrix}
 I_{14}& A_{14\times 2} & B_{14\times 8} \\
O_{5\times 14}& O_{5\times 2} & D_{5\times 8}
\end{pmatrix},
\end{align*}
where $I_{14}$ is the $14\times 14$ identity matrix, $A_{14\times 2}$ is determined by
\begin{align}
(A_{14\times 2})^{\top}=
\left (
\begin{array}{cccc cccc cccccc}
g^{43}& g^{24} & g^{59} & g^{20} &g^{77} &g^{25} &g^{28} &g^{46} &g^{30} &g^{60} &g^{30} &g^{70} &g^{60} &2\\
g^{34}& g^{69} & g^{30} & g^{7} &g^{35} &g^{38} &g^{56} &g^{53} &g^{70} &2 &1 &g^{70} &g^{50} &g^{50}
\end{array}
\right ),
\end{align}
$B_{14\times 8}=(O_{14\times 5},B'_{14\times 3})$, where
\begin{align}
(B'_{14\times 3})^{\top}=
\left (
\begin{array}{cccc cccc cccccc}
g^{21}& g^{42} & g^{31} & g^{17} &g^{77} &g^{49} &1 &g^{39} &g^{67} &g^{57} &g^{57} &g^{77} &g^{37} &g^{7}\\
g^{73}& g^{44} & g^{18} & g^{66} &g^{35} &g^{12} &0 &g^{79} &g^{27} &g^{17} &g^{17} &g^{37} &g^{77} &g^{47}\\
g^{67}& g^{53} & g^{77} & g^{18} &g^{78} &g^{22} &g^{50} &g^{12} &g^{37} &g^{27} &g^{27} &g^{47} &g^{7} &g^{57}
\end{array}
\right ),
\end{align}
and $D_{5\times 8}=(I_{5},D'_{5\times 3})$, where
\begin{align*}
(D'_{5\times 3})^{\top}=\begin{pmatrix}
g^{20}& g^{10} & g^{10} & g^{30} &g^{70}\\
g^{10}& g^{70} & g^{20} & g^{30} &g^{10}\\
g^{30}& g^{30} & g^{50} & g^{10} &g^{60}
\end{pmatrix}.
\end{align*}
Let $C=\{(i,\mathbf{c}_i):\;1\leq i\leq 24\}$, $C_1=\{(i,\mathbf{c}_i):\;1\leq i\leq 16\}$, and $U=\{(i,\mathbf{c}_{i}):\;21 \leq i \leq 24\}$. Then,
$$C=C_1\cup \{(17,\mathbf{c}_{17}),\ldots,(20,\mathbf{c}_{20})\}\cup U$$
constitutes a decomposition of Case (I) in Theorem \ref{Stru}.
\end{example}

\subsection{The third family of optimal quantum $(r,\delta)$-LRCs}

In the following theorem, we present the third family of optimal quantum $(r,\delta)$-LRCs.
\begin{theorem}\label{Opt-3}
Let $\mathbb{F}=\mathbb{F}_{q^{2}}$ with $q\geq 9$, and let $s,v,t\in \mathbb{N}^+$ satisfy
\begin{itemize}
\item  $v|(q-1),\;v\geq 8$ and $s\leq \lfloor \frac{v}{2}\rfloor-3$;

\item  $vt\leq q+v-s-4$.
\end{itemize}
Let $\mathcal{C}$ be a linear code with a parity-check matrix:
\begin{align}
H=\begin{pmatrix}
A_{1} & & &\\
& A_{2} &  &\\
&  & \ddots &\\
&  &  &A_t\\
B_{1}& B_{2} & \cdots &B_t
\end{pmatrix},  \label{qpari-2}
\end{align}
where
\begin{align}
A_i=
\left (
\begin{array}{cccc cccc cccc}
x_{i,0} & x_{i,1} & \cdots& x_{i,q-2}&y_0 &y_1&\cdots& y_{v-1}& &  & & \\
\vdots & \vdots & \cdots& \vdots&\vdots &\vdots&\cdots& \vdots& &  & & \\
x_{i,0}^s & x_{i,1}^s & \cdots& x_{i,q-2}^s&y_0^s &y_1^s&\cdots& y_{v-1}^s& &  & & \\
& & & &y_0 &y_1&\cdots& y_{v-1}&z_{i,0} & z_{i,1} & \cdots& z_{i,q-2}\\
& & & &\vdots &\vdots&\cdots& \vdots&\vdots & \vdots & \cdots& \vdots\\
& & & &y_0^s &y_1^s&\cdots& y_{v-1}^s&z_{i,0}^s & z_{i,1}^s & \cdots& z_{i,q-2}^s
\end{array}
\right )
\end{align}
and
\begin{align}
B_i=
\left (
\begin{array}{cccc cccc c}
x_{i,0}^{s+1} & \cdots& x_{i,q-2}^{s+1}&y_0^{s+1} & \cdots& y_{v-1}^{s+1}&z_{i,0}^{s+1} & \cdots& z_{i,q-2}^{s+1} \\
x_{i,0}^{s+2} & \cdots& x_{i,q-2}^{s+2}&y_0^{s+2} & \cdots& y_{v-1}^{s+2}&z_{i,0}^{s+2} & \cdots& z_{i,q-2}^{s+2}
\end{array}
\right ) \label{matrix-B}
\end{align}
for $i\in [t]$ and the blank entries in matrix $H$ are all zeros. Here,
\begin{numcases}{}
	 x_{i,j}=\lambda_i \omega_i^j,\; z_i=\mu_i \omega_i^j\;(i\in [t],\;0\leq j\leq q-2), \\
	y_l=\zeta^l\;(0\leq l\leq v-1),\label{ineq2}
\end{numcases}
where $\omega_1,\ldots, \omega_t\in \mathbb{F}$ are primitive $(q-1)$th roots of unity, $\zeta \in \mathbb{F}$ is a primitive $v$th root of unity,
and $\lambda_i,\mu_i\in \mathbb{F}$ such that $\{\lambda_i,\mu_i,\frac{\mu_i}{\lambda_i}\}\subseteq  \mathbb{F}^\times \backslash \{\omega_i^j:\;1\leq j\leq q-1\}$ for $i\in [t]$. Then, the following statements hold:
\begin{itemize}
\item [(1)] $\mathcal{C}$ is an optimal $(q+v-s-1,s+1)$-LRC with parameters $[t(2q+v-2),t(2q+v-2)-2st-2,s+3]_{q^2}$.

\item [(2)] $\mathcal{C}$ is Hermitian dual-containing, inducing an optimal quantum $(q+v-s-1,s+1)$-LRC with parameters
\begin{align}
\bigg[\mspace{-4mu}\bigg[t(2q+v-2),t(2q+v-2)-4st-4, s+3\bigg]\mspace{-4mu}\bigg]_{q}. \label{quan-3}
\end{align}
\end{itemize}
\end{theorem}

\begin{proof}
By Lemma \ref{local}, $\mathcal{C}$ is an $(q+v-s-1,s+1)$-LRC. For $i\in [t]$, we define
\begin{align}
A_{2i-1,1}=\begin{pmatrix}
x_0 & x_1 & \cdots& x_{q-2}\\
\vdots & \vdots & \cdots& \vdots\\
x_0^s & x_1^s & \cdots& x_{q-2}^s
\end{pmatrix},
\end{align}
\begin{align}
A_{2i-1,2}=A_{2i,1}=\begin{pmatrix}
y_0 & y_1 & \cdots& y_{q-2}\\
\vdots & \vdots & \cdots& \vdots\\
y_0^s & y_1^s & \cdots& y_{q-2}^s
\end{pmatrix},\ \
A_{2i,2}=\begin{pmatrix}
z_0 & z_1 & \cdots& z_{q-2}\\
\vdots & \vdots & \cdots& \vdots\\
z_0^s & z_1^s & \cdots& z_{q-2}^s
\end{pmatrix},
\end{align}
and set $(B_{2i-1},B_{2i})=B_i$ in Theorem \ref{Stru-apply}, where $B_i$ refers to the matrix given in Eq. \eqref{matrix-B}. To show that $\mathcal{C}$ is an optimal $(r,\delta)$-LRC, we only need to check the conditions of Theorem \ref{Stru-apply}. The only non-trivial condition is to prove $\lceil\frac{k}{r}\rceil=2t$ and $d(\mathcal{C})=s+3$. Since $k=t(2q+v-2)-2st-2=2t(q+v-s-1)-(2+vt)$ and $2+vt\leq q+v-s-2$ by assumption, we have $\lceil\frac{k}{r}\rceil=2t$. Clearly, the first $s+3$ columns of $H$ are linearly dependent. Thus, $d(\mathcal{C})\leq s+3$. Since $\{\lambda_i,\mu_i,\frac{\mu_i}{\lambda_i}\}\subseteq  \mathbb{F}^\times \backslash \{\omega_i^j:\;j\in [q-1]\}$ ensures that $x_{i,j}\neq y_l$, $x_{i,j}\neq z_{i,u}$ and $y_j\neq z_{i,u}$ for all $i\in [t],\; 0\leq j,u\leq q-2,\;0\leq l\leq v-1$, and by using the same technique as in Theorem \ref{Opt-2}, any $s+2$ columns of $H$ are linearly independent. Thus, $d(\mathcal{C})\geq s+3$. Since we have already shown $d(\mathcal{C})\leq s+3$, it follows that $d(\mathcal{C})=s+3$.

Next, we use Theorem \ref{Stru-quan} to prove that the induced quantum code from $\mathcal{C}$ is an optimal quantum $(r,\delta)$-LRC. For $1\leq i\leq 2st+2$, denote by $\mathbf{n}_{i}$ the $i$th row vector of $H$. Since $s\leq \lfloor \frac{v}{2}\rfloor-3$, and by computations analogous to those
in Theorem \ref{Opt-1}, all Hermitian inner products satisfy $\langle\mathbf{n}_{i},\mathbf{n}_{j}\rangle_{\mathrm{H}}=0$. Consequently, the codes whose generator matrices are given in Eq. \eqref{hermi} are Hermitian self-orthogonal. By Theorem \ref{Stru-quan}, $\mathcal{C}$ induces an optimal quantum $(r,\delta)$-LRC with parameters specified in Eq. \eqref{quan-3}. $\hfill\square$
\end{proof}

\begin{remark}\label{remark3}
The optimal quantum $(r,\delta)$-LRCs in Theorem \ref{Opt-3} satisfy $d(\mathcal{C})=\delta+2$, whereas those in Theorem \ref{Opt-2} attain $d(\mathcal{C})=\delta+1$.
This distinct difference in minimum distance ensures that there is no overlap between the two families of quantum $(r,\delta)$-LRCs.
\end{remark}

Let us give an example to illustrate Theorem \ref{Opt-3}.

\begin{example}\label{quan-example3}
Let $\mathbb{F}=\mathbb{F}_{11^2}$ with a primitive element $g$. Set $v=10,\;t=1,\;s=2$, and let $\omega,\zeta\in \mathbb{F}$ be primitive $10$th root of unity.
Choose $\lambda=g$ and $\mu=g^2$ in Theorem \ref{Opt-3}. Define
$$x_{1,i}=x_i=\lambda \omega^i,\;z_i=\mu \omega^i \ \mathrm{and} \ y_i=\zeta^i$$
for $0\leq i\leq 9$. Then, the code $\mathcal{C}$ in Theorem \ref{Opt-2} has a parity-check matrix:
\begin{align}
H=
\left (
\begin{array}{cccc cccc cccc}
x_0 & x_1 & \cdots& x_{9}&y_0 &y_1&\cdots& y_{9}&0 &0  & \cdots & 0\\
x_0^2 & x_1^2 & \cdots& x_{9}^2&y_0^2 &y_1^2&\cdots& y_{9}^2&0 & 0 &\cdots &0 \\
0&0 &\cdots &0 &y_0 &y_1&\cdots& y_{9}&z_0 & z_1 & \cdots& z_{9}\\
0&0 & \cdots &0 &y_0^2 &y_1^2&\cdots& y_{9}^2&z_0^2 & z_1^2 & \cdots& z_{9}^2\\
x_0^3 & x_1^3 & \cdots& x_{9}^3&y_0^3 &y_1^3&\cdots& y_{9}^3&z_0^3 & z_1^3 & \cdots& z_{9}^3\\
x_0^4 & x_1^4 & \cdots& x_{9}^4&y_0^4 &y_1^4&\cdots& y_{9}^4&z_0^4 & z_1^4 & \cdots& z_{9}^4
\end{array}
\right ).
\end{align}

By Theorem \ref{Opt-3}, $\mathcal{C}$ is an optimal $(18,3)$-LRC with parameters $[30,24,5]_{11^2}$, inducing an optimal quantum $(18,3)$-LRC with parameters $[\mspace{-2mu}[30,18,5]\mspace{-2mu}]_{11}$.

Furthermore, let us give a decomposition of Case (I) in Theorem \ref{Stru}. Let $G=(\mathbf{c}_1,\ldots,\mathbf{c}_{30})$ be a generator matrix of $\mathcal{C}$ satisfying $GH^\top=O_{24\times 6}$. Then, $G$ can be given by
\begin{align*}
\begin{pmatrix}
 I_{18}& A_{18\times 2} & B_{18\times 10} \\
O_{6\times 18}& O_{6\times 2} & D_{6\times 10}
\end{pmatrix},
\end{align*}
where $I_{18}$ is the $18\times 18$ identity matrix, $A_{18\times 2}$ is determined by
\begin{align*}
(A_{18\times 2})^{\top}=
\left (
\begin{array}{cccc cccc cccc cccc cc}
g^{27}& g^{18} & g^{77} & g^{59} &g^{52} &g^{46} &g^{33} &5 &g^{67} &g^{80} &8 &4 &4 &3& 8 &2 &3 &2\\
g^{42}& g^{101} & g^{83} & g^{76} &g^{70} &g^{57} &9 &g^{91} &g^{104} &g^{51} &5 &5 &1 &10& 8 & 1 & 8 &7
\end{array}
\right ),
\end{align*}
$B_{18\times 10}=(O_{18\times 6},B'_{18\times 4})$, where $(B'_{18\times 4})^{\top}$ is
\begin{align*}
\left (
\begin{array}{cccc cccc cccc cccc cc}
g^{44}& 10 & g^{8} & 1 &g^{65} &g^{117} &g^{10} &g^{4} &4 &g^{109} &g^{11} &g^{30} &g^{109} &g^{57}& g^{43} &g^{10} &2 &g^{62}\\
2& g^{94} & g^{41} & g^{46} &g^{13} &g^{10} &2 &g^{6} &g^{43} &7 &g^{51} &g^{54} &g^{43} &g^{116}& g^{95} & g^{89} & g^{28} &g^{25}\\
5& g^{88} & g^{101} & g^{3} &g^{39} &g^{33} &0 &g^{35} &g^{117} &g^{23} &g^{28} &0 &g^{16} &g^{112}& g^{4} & g^{112} & g^{28} &g^{64}\\
g^{45}& g^{89} & g^{42} & 4 &g^{93} &g^{30} &g^{26} &g^{76} &g^{54} &g^{16} &g^{32} &g^{6} &g^{39} &g^{89}& g^{57} & g^{98} & g^{55} &g^{35}
\end{array}
\right ),
\end{align*}
and $D_{6\times 10}=(I_{6},D'_{6\times 4})$, where
\begin{align*}
(D'_{6\times 4})^{\top}=\begin{pmatrix}
1& 8 & 4 & 6 &7&8\\
8& 10 & 7 & 8 &7&5\\
5& 4 & 8 & 4 &10&3\\
3& 7 & 5 & 4 &3&1
\end{pmatrix}.
\end{align*}
Let $C=\{(i,\mathbf{c}_i):\;1\leq i\leq 30\}$, $C_1=\{(i,\mathbf{c}_i):\;1\leq i\leq 20\}$, and $U=\{(i,\mathbf{c}_{i}):\;26 \leq i \leq 30\}$. Then,
$$C=C_1\cup \{(21,\mathbf{c}_{21}),\ldots,(25,\mathbf{c}_{25})\}\cup U$$
constitutes a decomposition of Case (I) in Theorem \ref{Stru}.
\end{example}

\section*{\large Declaration of competing interest}

The authors declare that they have no known competing financial interests or personal relationships that could have appeared to influence the work reported in this paper.

\section*{\large Acknowledgment}

K. Zhou is supported by the National Natural Science Foundation of China under Grant No. 12401040.
M. Cao is supported by the National Natural Science Foundation of China under Grant No. 12401684.

\section*{\large Data availability}

No data was used for the research described in this paper.

\end{document}